\documentclass[a4paper,UKenglish]{lipics}

\usepackage{tikz}
\usepackage{amssymb}

\usetikzlibrary{shapes}

\newcommand{\A}{\mathcal{A}}
\newcommand{\B}{\mathcal{B}}

\newcommand{\sa}{synchronizing automata}

\newcommand{\whp}{with probability
$1-O(\frac{1}{n})$}

\newtheorem{myremark}{Remark}[section]
\newtheorem{conjecture}[theorem]{Conjecture}

\usepackage{amsmath}
\usepackage{amssymb}
\usepackage{gastex}

\begin{document}

\title{On the probability of being synchronizable}



\author{Mikhail V. Berlinkov\ \thanks{\url{http://orcid.org/0000-0002-3903-0130}}}
%


\maketitle

\begin{abstract}
We prove that a random binary automaton with $n$ states is synchronizing with high probability and the asymptotic is exactly $1-\Theta(\frac{1}{n})$ as conjectured by Cameron~\cite{CamConj}. We also present a deterministic algorithm which decides whether a given random automaton is synchronizing in expected time linear in $n$ and prove that it is optimal. This is an expanded and heavily revised version of the conference paper~\cite{RandomAutCaldam} with an addition of the algorithmic result described above.
\keywords{synchronizing automata, random mappings, random digraphs}
\end{abstract}

\section{Synchronizing automata}
\label{intro}

Let $\A$ be a complete deterministic finite automaton whose input alphabet is $\Sigma$ and whose state set is $Q$.
$\A$ is called \emph{synchronizing} if there is a word $w\in \Sigma^*$ whose action \emph{resets} $\A$, that is, $w$ leaves the automaton in one particular state no matter at which state in $Q$ it is applied: $q.w=q'.w$ for all $q,q'\in Q$.
Any such word $w$ is called a \emph{reset word} of $\A$.
One can check that a word $(ab^3)^2a$ is a reset word for the automaton $\mathfrak{C}_4$ depicted on Figure~\ref{fig:underlying_graphs} (left).
This automaton belongs to \v{C}ern\'y's series of circular $n$-state synchronizing automata having shortest reset words of length $(n-1)^2$.
\v{C}ern\'y, who pioneered the study of synchronizing automata, introduced this series and conjectured that the series represents the extreme case in terms of the maximal length of reset words for $n$-state automata.
This statement became known as the \v{C}ern\'y conjecture.
\begin{conjecture}[\v{C}ern\'y, 1964]
Each synchronizing automaton with $n$ states has a reset word of length at most $(n-1)^2$.
\end{conjecture}
Despite attracting a lot of interest, the conjecture remains open, and the best known upper bound is cubic in $n$.
A cubic bound was first established in \cite{Fr1982,Pin1983} and also proved independently in~\cite{KRS1987}.
There are also a number of particular classes of automata for which more specific bounds were found.
For a brief introduction to the theory of \sa\ we refer the reader to the survey~\cite{Vo08}.
Synchronizing automata serve as transparent and natural models of error-resistant systems in many applications (coding theory, robotics, testing of reactive systems) and also reveal interesting connections with symbolic dynamics and other parts of mathematics.
Synchronization generally implies some desired property depending on the application.
For instance, if a prefix code is synchronizing, even if an error occurred during decoding, the correct decoding process will recover on its own with high probability (under certain assumptions on the input), thus making the process error-resistant.
In this paper we aim to show that synchronization is a universal property rather than an exception. Formally, we prove that a uniformly chosen random binary automaton is synchronizing with probability $1-O(\frac{1}{n})$ where $n$ is the number of states. Moreover, we prove that the convergence rate is tight thus confirming Cameron's conjecture~\cite{CamConj}. We believe that the proof can be easily generalized to show $1-O(\frac{1}{n^{0.5 k}})$ for any alphabet size $k>1$, though this might not be the exact asymptotic for $k>2$. 

Until 2016, the best results in this direction were much weaker: in~\cite{Zaks4} it was proved that random $4$-letter automata are synchronizing with probability $p$ for a specific constant $p>0$; in~\cite{Zaks10} it was proved that if a random automaton with $n$ states has at least $72 \ln(n)$ letters then it is almost surely synchronizing. Nicaud~\cite{FastSyn} showed (independently) by pure combinatorial techniques that a random $n$-state binary automaton is synchronizing with probability $1-O(n^{-\frac{1}{8}+o(1)})$ and admits a reset word of length $O(n^{1 + \varepsilon})$.
This implies that the \v{C}ern\'y conjecture holds with high probability.
Moreover, this result has been used in~\cite{BerlMarekAlgCrit} to show that the \v{C}ern\'y conjecture holds with probability exponentially close to $1$ in $n$ and  also that the expected value of the shortest reset word length of an $n$-state random synchronizing automaton is at most $n^{3/2+o(1)}$.

More recently, Chapuy and Perarnau~\cite{ChapuyPerarnau2023} established a major theoretical milestone by proving that random binary automata possess reset words of length $O(\sqrt{n} \log n)$ with high probability, which is a huge improvement and nearly matches existing experimental estimates around $O(\sqrt{n})$.

We devote the last section to a related algorithmic problem of deciding whether an automaton with $n$ states is synchronizable. The classic deterministic algorithm checks whether in the pair automaton $\mathcal{A}^2$ a diagonal set $\{ (q,q) \mid q \in Q \}$ is reachable from any (pair) state. This can be done by a BFS algorithm starting from the diagonal set following reversed edges in $O(n^2)$ time and space. While Nicaud~\cite{FastSyn} introduced a fast probabilistic linear-time heuristic for synchronization, we improve upon the exact deterministic baseline. We will leverage that the \emph{good path} (sufficient conditions for an automaton to be synchronizable) of the proof is verifiable in linear time $O(n)$ and the \emph{bad path} (when some of these conditions don't hold) happens with $O(1/n)$ probability (due to the main result), thus yielding an exact deterministic algorithm to verify that an automaton is synchronizing in $O(n)$ expected time. A modification of this algorithm was implemented and experimentally evaluated in advance of the present, complete write-up by Ageev~\cite{Ageev19}; here we give the full proof, including the matching $\Omega(n)$ lower bound showing this is the best computational complexity one can achieve. Finally, we will use this algorithm to experimentally support our main result.


\section{The probability of being synchronizable}
\label{sec_main}

Let $Q$ stand for $\{1,2, \dots ,n\}$ and $\Sigma_n$ for the probability space of all complete maps from $Q$ to $Q$ with the uniform probability distribution.
Let $\Sigma_{n}^{k}$ denote the probability space of $k$ independently chosen maps from $\Sigma_{n}$.
This probability space matches the probability space of all $k$-letter $n$-state automata with uniform distribution.
Throughout this paper let $\A=\langle Q,\{a,b\} \rangle$ be a random automaton from $\Sigma_{n}^{2}$, that is, $\A$ is chosen uniformly at random from the set of all $2$-letter automata with $n$ states.
As we do not have any restrictions on $\A$, we can choose such an automaton by picking $a$ and $b$ independently and uniformly at random (u.a.r.) from $\Sigma_n$, that is, $$P(\A = \langle Q,\{a',b'\} \rangle) = P_{a \in \Sigma_n}(a = a')P_{b \in \Sigma_n}(b = b').$$ In turn, the choice of a random mapping from $\Sigma_n$ is equivalent to choosing independently and u.a.r.
an image for each state $q \in Q$. Here and below by independence of two objects $O_1(\A)$ and $O_2(\A)$ determined by an automaton, we mean the independence of the corresponding events $O_1(\A) = o_1$ and $O_2(\A) = o_2$, for all instances $o_1, o_2$ from the corresponding sets.

We will also use the following shortcut concept. Let $O_1$ be a procedure (or partial function) that defines an object for an automaton $\A$. We say that $O_1(\A)$ is \emph{conditionally independent} of $O_2(\A)$, if $O_1(\A)$ is independent from $O_2(\A)$ provided $O_1$ is well defined. Formally, 
\begin{align*}
Pr(O_1(\A) = o_1; O_2(\A) = o_2 \mid O_1 \text{ is defined for } \A) &= \\  Pr(O_1(\A) = o_1 \mid O_1 \text{ is defined for } \A) \times Pr(O_2(\A) = o_2 \mid O_1 \text{ is defined for } \A). 
\end{align*}
Essentially, when we state a claim as `Suppose there exists $O(\A)$ independent of the letter ...; Then ...', we operate under the condition that $O(\A)$ is well defined, meaning $O(\A)$ is conditionally independent of that letter. Since we always operate under the assumption that all structures are well defined \whp, we could use a \emph{chain} argument to get the final bound.

For a procedure $O$ that defines an object $O(\A)$, denote by $\mathcal{D}(O)$ the event that $O$ is defined for $\A$. In this notation, $O_1$ is conditionally independent of $O_2$ if $O_1$ and $O_2$ are independent in the probability space $\bigl(\Sigma_n^2,\ Pr(\,\cdot \mid \mathcal{D}(O_1))\bigr)$. The following elementary facts make the chain argument precise.

\begin{myremark}
\label{rem:chain}
\emph{(Chain rule.)} Let $O$ be a procedure with $Pr(\mathcal{D}(O))>0$ and let $F$ be an event.
\begin{enumerate}
\item[(a)] $Pr(F ; \mathcal{D}(O)) \leq Pr(F \mid \mathcal{D}(O))$. Consequently, $Pr(F) \leq Pr(\neg\mathcal{D}(O)) + Pr(F \mid \mathcal{D}(O))$.
\item[(b)] If $O(\A)$ and $\mathcal{D}(O)$ are determined by a single letter $x \in \{a,b\}$, then $O$ is conditionally independent of the other letter $\bar x$.
\item[(c)] If $Pr(\mathcal{D}(O)) \geq 1/2$, then $Pr(F \mid \mathcal{D}(O)) \leq 2\,Pr(F)$.
\end{enumerate}
\end{myremark}
\begin{proof}
(a) $Pr(F ; \mathcal{D}(O)) = Pr(F \mid \mathcal{D}(O))\,Pr(\mathcal{D}(O))$ and $Pr(F) \leq Pr(\neg\mathcal{D}(O)) + Pr(F ; \mathcal{D}(O))$.

(b) Write $O(\A) = f(x)$ with $\mathcal{D}(O) = \{x \in \mathrm{dom}\, f\}$. For any object $o$ and any map $\beta$ for the other letter, since the letters are chosen independently,

\begin{align}
    Pr(f(x) = o;\ \bar x = \beta \mid \mathcal{D}(O)) = \frac{Pr(f(x) = o)\,Pr(\bar x = \beta)}{Pr(\mathcal{D}(O))} = \\ Pr(f(x) = o \mid \mathcal{D}(O))\,Pr(\bar x = \beta \mid \mathcal{D}(O)),
\end{align}
because $Pr(\bar x = \beta \mid \mathcal{D}(O)) = Pr(\bar x = \beta)$.

(c) $Pr(F \mid \mathcal{D}(O)) \leq Pr(F)/Pr(\mathcal{D}(O))$.
\end{proof}

A claim of the form `Suppose $O$ is independent of the letter $x$; Then a property $\Pi$ holds \whp' is a statement about the space $(\Sigma_n^2, Pr(\cdot \mid \mathcal{D}(O)))$, in which $O$ is independent of $x$ by definition. It says that $Pr(\neg\Pi \mid \mathcal{D}(O)) = O(1/n)$ and hence, by Remark~\ref{rem:chain}(a),
\begin{equation}
\label{eq:star}
Pr(\neg\Pi ; \mathcal{D}(O)) = O(1/n).
\end{equation}
We use \eqref{eq:star} under two rules.
\begin{enumerate}
\item \emph{Intersect, never condition.} The independence hypothesis holds only in the space conditioned on $\mathcal{D}(O)$. Therefore, we never condition on any other event, and in particular not on events such as `the pair $st(a)$ is stable', which depend on both letters. Such events enter only through \emph{deterministic implications}, after which the event under consideration is enlarged to one that mentions only $x$ and $O$, to which \eqref{eq:star} applies.
\item \emph{Failure of well-definedness is measured relative to the previous structure.} If $O$ is built from a structure $O'$, the relevant bound is $Pr(\mathcal{D}(O') ; \neg\mathcal{D}(O)) \leq Pr(\neg\mathcal{D}(O) \mid \mathcal{D}(O')) = O(1/n)$, not $Pr(\neg\mathcal{D}(O))$. Indeed, $Pr(\mathcal{D}(st(a)))$ is only $1 - \Theta(n^{-1/2})$ by Theorem~\ref{th_high_tree}.
\end{enumerate}

Our main result is the following theorem.
\begin{theorem}
\label{th_main} The probability of being synchronizable for $2$-letter random automata with $n$ states equals
$1-\Theta(\frac{1}{n})$.
\end{theorem}
 
\section{Connectivity and the Upper Bound}
\label{sec:upper_bound}

Let us define a \emph{subautomaton} as a terminal strongly-connected component (SCC)~\footnote{A strongly-connected component $S$ is terminal 
when $S\cdot u\subseteq S$ for every $u\in \Sigma^*$.} of $\A$.
Observe that if an automaton is synchronizing, it must have a unique terminal SCC. Hence the following lemma gives a matching lower bound of $\Omega(1/n)$ on the probability that $\A$ is not synchronizing, establishing the upper bound $1-\Omega(1/n)$ on the probability that $\A$ is synchronizing claimed in Theorem~\ref{th_main}.
\begin{lemma}
\label{lem_weak} The probability that $\A$ does not have a unique
terminal SCC is at least $\Omega(\frac{1}{n})$.
\end{lemma}
\begin{proof}
Let us count the number of automata having exactly one
\emph{disconnected loop}, that is automata having a state with only (two) incoming arrows from itself.
We first choose a unique state $p$ of a disconnected loop in $n$ ways.
The transitions for this state are defined uniquely.
For any other state $q$, we define transitions in $(n-2)^2$ ways by choosing the image for each letter to be any state except $p,q$, ensuring that $q$ is not disconnected and $p$ has no transitions from other states.
Thus the probability of having exactly one disconnected loop is at least
$$\frac{n (n-2)^{2(n-1)}}{n^{2n}} = \frac{1}{n}\left(1-\frac{2}{n}\right)^{2(n-1)} = \Theta\left(\frac{1}{n}\right).$$
This concludes the proof of the lemma as such automata have multiple SCCs.
\end{proof}



\begin{lemma}[Concentration of Giant Subautomaton Size \cite{Grusho1973, CaiDevroye2017}]
\label{lem_size_of_subaut}
Let $\mathcal{A}$ be a random automaton on $n$ states with alphabet size $k=2$. Then, with probability $1 - O(1/n)$, $\mathcal{A}$ contains a unique terminal strongly connected component $\mathcal{C}$. Let $s$ denote the state count of its unique terminal strongly connected component, then $s$ is concentrated around $\nu_2 n$ (where $\nu_2 \approx 0.7968$), and for any $\alpha < \nu_2$, there exists a constant $c = c(\alpha) > 0$ such that:
\[
\mathbb{P}\left(s \le \alpha n\right) \le e^{-c n}.
\]
In particular, for $\alpha = 3/4$, $\mathbb{P}\left(s < 3n/4\right) = O(e^{-\Omega(n)})$.
\end{lemma}

\section{The Lower Bound}
\label{sec:lower_bound}

Now we turn to the proof of the lower bound of Theorem~\ref{th_main} by means of a top-down approach. In order to describe the plan, we need a few definitions which we elaborate on further in the proof. First, call a set of states $K \subseteq Q$ \emph{synchronizable} if it can be mapped to one state by
some word. Next, a pair of states $\{p,q\}$ is called \emph{stable} if it cannot be mapped by any word into a non-synchronizable pair. Finally, the \emph{underlying graph} of a letter (or mapping) $a$ is a digraph $\Gamma_a = \Gamma(Q, \{(q, q.a) \mid q\in Q\})$.
Since it has common out-degree $1$, the underlying graph consists of one or more maximal weakly connected components~\footnote{a subgraph, maximal by inclusion, containing a state accessible from every state of the subgraph} called \emph{clusters}, each one consisting of a unique cycle and trees rooted on this cycle (see Figure~\ref{fig:underlying_graphs}).
The plan is as follows. After recalling necessary well known inequalities in Subsection~\ref{sub_sec:useful}, we state an upper bound on the number of clusters of a random mapping in Subsection~\ref{sub_sec:cluster_structure}, which will be used throughout the paper.
Then, in Subsection~\ref{sub_sec:synch_sets} we prove that having the set of \emph{big} clusters for each letter synchronizable is enough to prove the lower bound.
Next, in Subsection~\ref{sub_sec:stable_pairs} we will show that this property can be provided by having a big enough set of \emph{stable pairs} for each letter independent of the other letter.
In Subsection~\ref{sub_sec:extending_stable_pairs} we prove that these sets of stable pairs can be built from just one stable pair independent of one of the letters.
In Subsection~\ref{sub_sec:one_stable_pair} we will show 
that if the underlying graph of one of the letters has a unique highest \emph{$1$-branch}, then there is a pair of states completely defined by this letter which is additionally stable provided the \emph{crown} of the $1$-branch intersects with any \emph{subautomaton}.
Finally, in Subsection~\ref{sub_sec:highest_branch} we prove that with high probability one of the letters of a random automaton has a unique highest $1$-branch whose crown is big enough to be reachable from any subautomaton with high probability.
\subsection{Useful Asymptotics}
\label{sub_sec:useful}

First, let us recall a few asymptotic equations that will be used throughout the paper. Stirling's approximation formula states that for $k>0$
\begin{equation}
\label{eq:stirling}
\tag{st}
k! = \left(\frac{k}{e}\right)^k \sqrt{2\pi k}(1+o(1)).
\end{equation}
Using Taylor's expansion of a logarithm $\ln(1-x) = -\sum_{i=1}^{+\infty}\frac{x^i}{i}$, for $0 < k < m$ we get
\begin{equation}
\label{eq:exp1}
\tag{exp1}
\left(1-\frac{k}{m}\right)^m = \exp\left({- \sum_{i=1}^{+\infty}\frac{k^i}{i m^{i-1}}}\right).
\end{equation}
\begin{equation}
\label{eq:exp2}
\tag{exp2}
e^{-2k} \leq \left(1-\frac{k}{m}\right)^m \leq e^{-k} \text{ if } k \leq \frac{m}{2}.
\end{equation}
\begin{equation}
\label{eq:exp3}
\tag{exp3}
\left(1-\frac{k}{m}\right)^m = e^{-k(1+o(1))} \text{ if } k = o(m).
\end{equation}
\begin{equation}
\label{eq:exp4}
\tag{exp4}
\left(1-\frac{k}{m}\right)^m = e^{-k}(1+o(1)) \text{ if } k^2 = o(m).
\end{equation}

Using (\ref{eq:stirling}), for the number of combinations for $0 < k < m$ we have
\begin{multline}
\label{eq:comb}
\tag{cmb0}
{m \choose k} = \frac{m!}{(m-k)!k!} = \frac{\Theta(1)\left(\frac{m}{e}\right)^{m+1/2}}{\left(\frac{k}{e}\right)^{k+1/2} \left(\frac{m-k}{e}\right)^{m-k+1/2}} = \\ = \frac{\Theta(1){m}^{m}}{{k}^{k}(m-k)^{m-k}}\sqrt{\frac{m}{k(m-k)}} = \frac{\Theta(1){m}^{k}}{{k}^{k}(1-\frac{k}{m})^{m-k}}\sqrt{\frac{m}{k(m-k)}}.
\end{multline}
If $k = o(m)$, then due to (\ref{eq:exp3}), (\ref{eq:comb})  simplifies to 
\begin{equation}
\label{eq:comb1}
\tag{cmb1}
{m \choose k} = \frac{\Theta(1){\left(m e^{1+o(1)}\right)}^{k}(1-\frac{k}{m})^{k}}{{k}^{k}}\sqrt{\frac{1}{k}} = \frac{\Theta(1){\left(m e(1+o(1))\right)}^{k}}{{k}^{k}\sqrt{k}},
\end{equation}
and if $k^2 = o(m)$, due to (\ref{eq:exp4}) it further simplifies to
\begin{equation}
\label{eq:comb2}
\tag{cmb2}
{m \choose k} = \frac{\Theta(1){\left(m e^{1+o(1)}\right)}^{k}(1-\frac{k}{m})^{k}}{{k}^{k}}\sqrt{\frac{1}{k}} = \frac{\Theta(1){\left(m e\right)}^{k}}{{k}^{k}\sqrt{k}}.
\end{equation}
In the general case $0 < k < m$, as $m \leq O(1)k(m-k), (1-\frac{k}{m}) \leq 1$ and due to (\ref{eq:exp2}) we can upper bound (\ref{eq:comb}) as follows.
\begin{equation}
\label{eq:comb3}
\tag{cmb3}
{m \choose k} \leq \frac{\Theta(1){m}^{k}}{{k}^{k}(1-\frac{k}{m})^{m-k}}\sqrt{\frac{m}{k(m-k)}} \leq  \frac{O(1){m}^{k}e^{2k}}{{k}^{k}} = O(1)\left(\frac{e^2m}{k}\right)^{k}.
\end{equation}
Notice that for $k \in \{0,m\}$, (\ref{eq:comb3}) also holds if we assume $0^0 = 1$.
\subsection{The Cluster Structure of Underlying Graphs}
\label{sub_sec:cluster_structure}

An example of an automaton with $4$ states and the underlying graphs of its letters is given in Figure~\ref{fig:underlying_graphs} -- the underlying graph of $b$ (on the right) has only one cluster, while the underlying graph of $a$ (in the middle) consists of $3$ clusters having the sets of vertices $\{0,1\}$, $\{2\}$, $\{3\}$, correspondingly.
Clearly, each directed graph with $n$ vertices and constant out-degree $1$ corresponds to a unique map from $\Sigma_n$.
Thus we can consider $\Sigma_n$ as the probability space with the uniform distribution on all directed graphs with constant out-degree $1$.
\begin{figure}[h]
\begin{center}
\begin{tikzpicture}[scale=1.5]

\begin{scope}
\node[draw,circle] (p0) at (0,1) {0};
\node[draw,circle] (p1) at (1,1) {1};
\node[draw,circle] (p2) at (1,0) {2};
\node[draw,circle] (p3) at (0,0) {3};
\draw[->,thick] (p0) edge node[above]{$a,b$} (p1);
\draw[->,thick] (p1) edge[loop above] node[above]{$a$} (p1);
\draw[->,thick] (p2) edge[loop below] node[below]{$a$} (p1);
\draw[->,thick] (p3) edge[loop below] node[below]{$a$} (p1);

\draw[->,thick] (p1) edge node[right]{$b$} (p2);
\draw[->,thick] (p2) edge node[below]{$b$} (p3);
\draw[->,thick] (p3) edge node[left]{$b$} (p0);
\end{scope}

\begin{scope}[xshift=2cm]
\node[draw,circle] (p0) at (0,1) {0};
\node[draw,circle] (p1) at (1,1) {1};
\node[draw,circle] (p2) at (1,0) {2};
\node[draw,circle] (p3) at (0,0) {3};

\draw[->,thick] (p0) edge node[above]{} (p1);
\draw[->,thick] (p1) edge[loop above] node[above]{} (p1);
\draw[->,thick] (p2) edge[loop below] node[below]{} (p1);
\draw[->,thick] (p3) edge[loop below] node[below]{} (p1);
\end{scope}

\begin{scope}[xshift=4cm]
\node[draw,circle] (p0) at (0,1) {0};
\node[draw,circle] (p1) at (1,1) {1};
\node[draw,circle] (p2) at (1,0) {2};
\node[draw,circle] (p3) at (0,0) {3};

\draw[->,thick] (p0) edge node[above]{} (p1);
\draw[->,thick] (p1) edge node[right]{} (p2);
\draw[->,thick] (p2) edge node[below]{} (p3);
\draw[->,thick] (p3) edge node[left]{} (p0);
\end{scope}

\end{tikzpicture}

\end{center}
\caption{Left to right: an automaton $\mathfrak{C}_4$, the underlying graphs of its letters $a$ and $b$.\label{fig:underlying_graphs}}
\end{figure}

The number of clusters (or cycles) of a random mapping is known to be concentrated around $\frac{\ln{n}}{2}$ and accurate bounds were established for various rates of growth of the number of clusters (see e.g.~\cite{TIM11}).
We need a stronger upper bound for large deviations of the number of clusters.
We believe the following statement can be proved by adapting the proof from~\cite{TIM11} to yield an inequality for the general case.
\begin{lemma}[Nicaud, 2019]
\label{lem_cnt_clusters}With probability
$1-o(1/{n^4})$, a random digraph from $\Sigma_n$ has at most $5\ln{n}$ clusters.
\end{lemma}
\begin{proof}
We will make use of powerful and beautiful theory of Analytic Combinatorics developed in \cite{AnalyticComb}.
Let $t_{n,k}$ be the number of mappings from $[n]$ to $[n]$ having $k$ clusters (component).
The associated exponential bivariate generating series defined by
\[
T(z,u) := \sum_{n\geq 0}\sum_{k\geq 0} \frac{t_{n,k}}{n!}z^nu^k
\]
can be written, using the formal method (see e.g.~\cite{RandMappings}):
\begin{equation}
\begin{cases}
\label{eq:Tzu}
T(z,u) = \exp\left(u\log\left(\frac1{1-C(z)}\right)\right), \\
C(z)  = z\exp(C(z))
\end{cases}
\end{equation}
As every coefficient ($t_{n,k}$) of the development of $T(z,u)$ near $(z,u)=(0,0)$ is non-negative, we have, for every $0<\rho<e^{-1}$ and every positive $\tau$:
\[
T(\rho,\tau) = \sum_{n,k}\left([z^nu^k] T(z,u)\right)\rho^n \tau^k \geq \left([z^nu^k] T(z,u)\right)\rho^n \tau^k,
\]
for every $n$ and every $k$.
Hence (this is called a saddle-point bound):
\begin{equation}
\label{eq:tnk_upper_bound}
\frac{t_{n,k}}{n!} = [z^nu^k] T(z,u) \leq \frac{T(\rho,\tau) }{\rho^n \tau^k}
\end{equation}
This holds for any $0<\rho<e^{-1}$, any positive $u$ and any $n$ and $k$.
We choose
\[
\rho = e^{-1}\left(1-\frac1{2n}\right);\quad k := \lambda_{n} + i;\quad \lambda_{n} := \lceil \lambda \log n\rceil,
\]
for any $i \geq 0$ and some $\lambda>0$ to be fixed later.
Due to~(\ref{eq:tnk_upper_bound}), the probability that a random mapping with $n$ states has $k$ components is
\begin{equation}
\label{eq:tnk_original_bound}
\frac{t_{n,k}}{n^n} \leq \frac{n!}{n^n}\frac{T(\rho,\tau) }{\rho^n \tau^k}.
\end{equation}
First observe that, due to (\ref{eq:stirling}) and the choice of $\rho$,
\begin{equation}
\label{eq:exp_part_upper_bound}
\frac{n!}{n^n \rho^n } = \frac{\sqrt{2\pi n}(1+o(1))}{\left(1-\frac{1}{2n}\right)^n} =  \frac{\sqrt{2\pi n}}{e^{-1/2}}(1+o(1)).
\end{equation}
Now we estimate $T(\rho,\tau)$. We use the classical development near $z=e^{-1}$ (see ~\cite{RandMappings})
\[
C(z) = 1 - \sqrt{2(1-ez)} + O(1-ez),
\]
and thus $1-C(\rho) = \frac1{\sqrt{n}} + O(\frac{1}{n})$, and thus from (\ref{eq:Tzu}) we have
\begin{equation}
\label{eq:Trhotau_bound}
T(\rho,\tau) =  \exp\left(\tau\log{\frac{\sqrt{n}}{1+O(1/\sqrt{n})}}\right)  = n^{\tau/2(1+o(1/n^{1/2}))}.
\end{equation}
Putting together (\ref{eq:tnk_original_bound}),(\ref{eq:exp_part_upper_bound}) and (\ref{eq:Trhotau_bound}), for some positive $\alpha$, 
\[
\frac{t_{n,k}}{n^n}  \leq \alpha \sqrt{n}\ \frac{n^{\tau/2}}{n^{\lambda\log\tau}} \tau^{-i}
= \alpha\tau^{-i} n^{\frac{1}{2}(\tau+1) - \lambda \log \tau} 
\]
If we choose $\tau = \lambda=5$, we have $\frac12(\tau+1)-\lambda\log\tau \leq -5$,
so that
\[
\frac{t_{n,k}}{n^n}  \leq \alpha n^{-5} 5^{-i},
\]
and therefore
\[
\mathbb{P}_{g \in \Sigma_n}(g \text{ has more than }5\log n\text{ clusters})
= \sum_{i\geq 0} \frac{t_{n,\lceil 5\log n \rceil +i}}{n^n}
\leq \frac{\beta}{n^5},
\]
for some positive $\beta$.
\end{proof}

\subsection{Independent Synchronizable Sets}
\label{sub_sec:synch_sets}

A set of states $K \subseteq Q$ is \emph{synchronizable} if it can be mapped to one state by some word.
In contrast, a pair of states $\{p,q\}$ is called a \emph{deadlock} if $p.s \neq q.s$ for each word $s$.
Notice that a deadlock pair remains deadlock under the action of any word.
First we aim to show that for proving that $\A$ is
synchronizing with probability $1-O(\frac{1}{n})$, it is enough to show that \whp\ for each letter there is a large synchronizable set of states which is completely defined by this letter.
Given $x \in \{a,b\}$ and a constant $1/3 < \alpha_c < 1$ to be specified later, let $S_x$ be the set of \emph{big} clusters of $\Gamma_x$ -- of size more than $n^{\alpha_c}$ states -- and let $T_x$ be the complement of $S_x$, or equivalently, $T_x$ be the set of \emph{small} clusters -- the clusters containing at most $n^{\alpha_c}$ states.
Since $S_x$ and $T_x$ are completely defined by $x$, both are independent of the other letter.
Due to Lemma~\ref{lem_cnt_clusters}, \whp\ there are at most $5\ln{n}$ clusters in $\Gamma_x$, hence $T_x$ contains at most $5\ln{(n)}n^{\alpha_c}$ states \whp.
Given a set of clusters $X$, denote by $\widehat{X}$ the set of states in the clusters of $X$.
\begin{theorem}
\label{thm:big_sync_sets}If $\widehat{S_b}$ is synchronizable \whp\ for a random automaton $\A$, then $\A$ is synchronizing \whp.
Formally, $$Pr(\widehat{S_b} \text{ is synchronizable}) = 1-O\left(\frac{1}{n}\right) \rightarrow Pr(\A \text{ is synchronizable}) = 1-O\left(\frac{1}{n}\right).$$
By the symmetric argument (swap $a,b$ throughout the proof below), the same holds with $\widehat{S_a}$ in place of $\widehat{S_b}$; we use whichever form matches the letter available in a given application.
\end{theorem}
\begin{proof}

\begin{lemma}
\label{lem:two_subsets}
Let $r$, $s$ and $k$ be three positive integers such that $r, s \leq n$ and $k \leq r$.
Let $R$ (resp. $S$) be a subset of Q of size $r$ (resp. $s$) taken uniformly at random.
If $S$ and $R$ are independent, then
\begin{equation}
Pr(|R \cap S| \geq k) \leq {{s \choose k}{n-k \choose r-k}}/{{n \choose r}}.
\end{equation}
In particular, if $|R| = O(1)$
\begin{equation}
\label{eq:two_subsets_const}
    Pr(|R \cap S| \geq k) = O(({s}/{n})^k).
\end{equation}
\end{lemma}
\begin{proof}
As $R$ is independent of $S$, we may assume that $S$ is fixed and we choose $R$ at random.
First we choose some $k$ states from $S$ in ${|S|
\choose k}$ ways, then we choose the rest $|R|-k$ states of $R$ from the remaining $n-k$ states in ${n-k \choose |R|-k}$ ways, and finally divide by ${n \choose |R|}$ -- the total number of ways to choose $R$ from $Q$.
This is an upper bound because if the overlap is larger than $k$, there are different choices of overlap that lead to the same choice of $R$.
(\ref{eq:two_subsets_const}) easily follows from (\ref{eq:comb2}).
\end{proof}

Now let us upper bound the probability that $\A$ is not synchronizing.
\begin{myremark}
\label{rem:assumptions}
Since $\widehat{S_b}$ is synchronizable \whp, we have
\begin{equation}
Pr(\A \text{ is not synch.}) = Pr(\A \text{ is not synch.}; \widehat{S_b} \text{ is synch.}) + O\left(\frac{1}{n}\right).
\end{equation}
This means we can instead upper bound the probability of the joint event on the right side or, in other words,
assume that $\widehat{S_b}$ is synchronizable in the proof. (Only $\widehat{S_b}$ is used anywhere below -- $\widehat{S_a}$'s synchronizability is never needed for this, $a$-cycle-reduction, direction of the argument.)
We can do the same with the condition that $\A$ has at most $5\ln{n}$ clusters for $a$ and $b$, hence assume that both $\widehat{S_a}$ and $\widehat{S_b}$ are $n - o(n)$.
\end{myremark}

If $\A$ were not synchronizing, a deadlock pair $\{p,q\}$ would exist.
We could also assume $\{p,q\}$ belongs to $a$-cycles, since it can always be mapped to cycles.
Given a state $x$ on $a$-cycle, denote by $cycle(x)$ an $a$-cycle containing $x$ and by $s_x$ its size.
In what follows, we show that one of configurations has to appear for such a pair $\{p,q\}$ to exist and upper bound their probabilities.
To simplify notations, denote $\beta = \lceil \frac{2-\alpha_c}{1-\alpha_c} \rceil$. First, we differentiate by whether $p, q$ both belong to the same $a$-cycle and the length of that cycle.

\textbf{1. One cycle of a big size.} $cycle(p) = cycle(q), s_p \geq 2\beta$.
Assuming $\widehat{S_b}$ is synchronizing, each deadlock pair must have at least one state in $\widehat{T_b}$.
Let us consider all pairs that can be obtained from $\{p,q\}$ on the cycle by applying $a$, that is, $\{p,q\}.a^{i}$ for $i=0, 1, \dots, s_p-1$.
In each of these $s_p$ pairs, at least one state must be in $\widehat{T_b}$ and each state on the cycle appears in exactly two pairs.
It follows that at least half of the states from $cycle(p)$ must belong to $\widehat{T_b}$.
We can upper bound the probability of such configuration as follows~\footnote{Note that a configuration itself does not depend on $\{p,q\}$, the pair is only used to show what configurations are possible.}.
First we choose an $a$-cluster in at most $5\ln{n}$ ways, which also defines an $a$-cycle, $c_p$ of size $s_p = |c_p|$.
Then we choose $z$ in at most $s_p$ ways and a subset $I \subseteq c_p$ of $z$ states in ${s_p \choose z}$ ways.
Since this subset is now completely defined by the letter $a$, the probability that $I \subseteq \widehat{T_b}$ is ${|\widehat{T_b}|
\choose z}/{n \choose z}$ (since $\widehat{T_b}$ is independent of $a$).
Putting all together, the overall probability of such configuration is upper bounded by
\begin{equation}
\label{eq:one_cycle1}
   \max_{s_p = 2\beta}^{n}\max_{z = \lceil \frac{s_p}{2} \rceil}^{s_p} {5(\ln{n}) s_p {s_p \choose z}  {|\widehat{T_b}|
\choose z}/{n \choose z}} \leq \max_{s_p = 2\beta}^{n}\max_{z = \lceil \frac{s_p}{2} \rceil}^{s_p} {5(\ln{n}) s_p 2^{s_p} {|\widehat{T_b}| \choose z}/{n \choose z}}.
\end{equation}
Note that we don't multiply by the number of ways to choose $s_p$ but instead taking the maximum, as $s_p$ is defined by the choice of the cluster $c_p$.
Using that $|\widehat{T_b}| = o(n)$, one can easily see that the expression under the innermost max operator in (\ref{eq:one_cycle1}) decreases by $z$, and thus the maximum is achieved for $z = \lceil \frac{s_p}{2} \rceil$ and is equal to
\begin{equation}
f_1(s_p) = {5(\ln{n}) s_p 2^{s_p} {|\widehat{T_b}|
\choose \lceil \frac{s_p}{2} \rceil}/{n \choose \lceil \frac{s_p}{2} \rceil}}.
\end{equation}
Similarly, one can see that $f_1(s_p)$ is also decreasing by $s_p$ for $n$ big enough.
Thus the maximum is achieved for $s_p = 2\beta$ and is upper bounded by 
\begin{equation}
O(1)\ln{n} {5(\ln{n})n^{\alpha_c} \choose \beta}/{n \choose \beta} = O(1)\frac{\ln^{\beta+1}{n}} {n^{\beta (1 - \alpha_c )}} = o(1/n).
\end{equation}

\textbf{2. Two cycles of big size.} $cycle(p) \neq cycle(q), s_p \geq 2\beta, s_q \geq 2\beta$.
This case is similar to the one big cycle case as we again define a configuration with a sufficiently big subset of states in $\widehat{T_b}$ and independent of $b$. However, we will intentionally overestimate the probability (of the upper bound) for the purpose of simplifying the algorithm in Section~\ref{sec:algo}.  

Denote by $c_{r,i}$ the $i$-th state on the cycle $cycle(r)$ for some order induced by the cycle, i.e. $c_{r,i}.a = c_{r,i+1 \bmod s_r}$.
Let $d$ be the g.c.d. of $s_p$ and $s_q$. Then for some $0 \leq x < d$ and all $k_1 > 0, k_2>0$ and $ i \in  \{0,1,\dots,d-1\}$, the pairs
\begin{equation}
\label{eq_ddlck_pairs} \{c_{p,(i + k_1 d) \bmod s_p},
c_{q,(x + i + k_2 d) \bmod s_q}\} \text{ are deadlocks}
\end{equation}
because we can get them as images of $\{p,q\}$.
It follows that in each of these pairs at least one of the states is in $\widehat{T_b}$ (or $\widehat{S_b}$ is not synchronizing which happens with $O(1/n)$).
Since $k_1,k_2$ are arbitrary in~(\ref{eq_ddlck_pairs}), $\{0,1,\dots,d-1\}$ can be split into two subsets $I_p$ and $I_q$ such that $c_{p,(i + k_1 d) \bmod s_p} \in \widehat{T_b}$ for each $i \in I_p$ and $c_{q,(x + i + k_2 d) \bmod s_q} \in \widehat{T_b}$ for each $i \in I_q$ for all non-negative integers $k_1, k_2$.

Let us define $J_p = \{ c_{p,(i + k_1 d) \bmod s_p} \mid i \in I_p\}$ and $J_q = \{c_{q,(x + i + k_2 d) \bmod s_q} \mid i \in I_q\}$. It follows from the definition that $J_p \subseteq \widehat{T_b}$ and $J_q \subseteq \widehat{T_b}$. Note also that $\widehat{J}_p = J_p \cap \{ c_{p,j} \mid 0 \leq j  \leq s_p-1\}$ has size $k s_p / d$ and $\widehat{J}_q = J_q \cap \{c_{q,j} \mid 0 \leq j \leq s_p-1\}$ has size $s_p(1 - k / d)$.

\begin{itemize}
    \item $25\ln^2{n}$ ways to choose clusters $cycle(p) \neq cycle(q)$;
This choice determines their sizes $s_p, s_q$ as well as their g.c.d. $d = \gcd{(s_p, s_q)}$;
    \item We sum over $1 \leq j_p \leq s_p$, the number of states in $\widehat{J}_p \cap \widehat{T_b}$.
    \item For each $j_p = 0, 1, \dots, s_p-1$ there are ${s_p \choose j_p}$ ways to choose a subset $\widehat{J}_p \subseteq c_p$ of size $j_p$ and ${s_p \choose s_p - j_p}$ ways to choose a subset $\widehat{J}_q \subseteq c_q$ of size $s_p - j_p$.
\end{itemize}

Since $\widehat{T_b}$ is independent of $a$ and these sets are completely defined by $a$ and our choice, the probability that the corresponding states from $\widehat{J}_p \cup \widehat{J}_q$  belong to $\widehat{T_b}$ equals ${|\widehat{T_b}| \choose s_p}/{n \choose s_p}$ (this is a particular case of Lemma~\ref{lem:two_subsets}).
Thus the probability of such configuration for $s_p,s_q$ being the cycle lengths of chosen cycles and $k$ being fixed, is at most
\begin{equation}
\label{prob_sb}  25\ln^2({n}) \max_{s_p,s_q, s_p \leq s_q} \sum_{j_p=0}^{s_p-1} {s_p \choose j_p}{s_p \choose s_p-j_p}{{|\widehat{T_b}| \choose s_p}}/{{n \choose s_p}}.
\end{equation}

Due to Vandermonde's identity, $\sum_{j_p=0}^{s_p-1} {s_p \choose j_p}{s_p \choose s_p-j_p} \leq {2s_p \choose s_p} \leq 2^{2s_p}$. Hence (\ref{prob_sb}) is bounded as
\begin{equation}
\label{eq_fk}
25\ln^2({n}) \max_{s_p,s_q, s_p \leq s_q} 2^{2s_p}{{|\widehat{T_b}| \choose s_p}}/{{n \choose s_p}}.
\end{equation}

Since $|\widehat{T_b}| = O(n^{\alpha_c})$, simple analysis would show that the maximum of (\ref{eq_fk}) is reached for the smallest possible value of $s_p$.
Since $s_p \geq 2\beta$, by (\ref{eq:comb2}) we get
\begin{equation}
    O(\ln^2{n})\left(\frac{|\widehat{T_b}|}{n}\right)^{2\beta} \leq O( 1 / n^{2\beta(1-\alpha_c)}) \leq O(1/n^{2-\alpha_c}).
\end{equation}
Hence (\ref{prob_sb}) is $o(1/n)$ as required.

\textbf{3. One of the cycles is small.} $s_p < 2\beta$.
We cannot use the same argument as when both cycles are big, because we cannot obtain a big enough overlap with $\widehat{T_b}$ (even if one of the states belong to a big cycle, it may happen that a smaller cycle is contained in $\widehat{T_b}$ while the bigger one has no overlap).
Note that we can `make' $\{p, q\}$ independent of a letter $b$ as follows.
First, we can assume $q$ is the state with the smallest index on its $a$-cycle (otherwise, we can map $\{p,q\}$ using letter $a$ for this to be true).
Second, we consider $\{p, q\}$ to be any such possible pair with $s_p < 2\beta$ and $q$ having the smallest index on its $a$-cycle.
By doing this, $\{p, q\}$ is defined by a pair of $a$-cycles and a choice of a state on an $a$-cycle (containing $p$), which makes it independent on the choice of $b$.
With $1-o(1/n)$ probability, there are at most $O(2\beta \ln^2{n})$ of such pairs, and we can use the union probability bound, once we establish that the probability for any of such pairs being a deadlock is $o(1/(n\ln^2{n}))$.

\textbf{3.1. No collisions.}
Let $c_p, c_q$ be two $a$-cycles containing states $p, q$ resp. as defined above.
Consider the chain of states $(p, q).ba^i$ for $i = 1, \dots, 2\beta$ and let $D$ be the set of these states.
Let us upper bound the probability that all the states in $D$ are distinct and not equal to $p, q$ (but might have overlap with cycles $c_p, c_q$), and for each $i > 0$, $|\widehat{T_b} \cap \{p, q\}.ba^i|
> 0$. Note that if the latter condition doesn't hold, the pair $\{p, q\}$ cannot be a deadlock.
The set $D$ is independent of the letter $b$ conditioned on not having overlap with $\{p, q, p.b, q.b\}$.
Indeed, $D$ is defined by $a$-transitions only conditioned upon being distinct from  $\{p, q, p.b, q.b\}$.
Note that since we have defined some transitions by letter $b$, we cannot assume it is random.
However, Remark~\ref{rem:assumptions} still holds. Since $\{p, q, p.b, q.b\}$ is just a constant number and $|\widehat{S_b}|
= n-o(n)$, the conditioning doesn't change the asymptotic below (the worst case is when all $\{p, q, p.b, q.b\}$ are in $\widehat{S_b}$).
By Lemma~\ref{lem:two_subsets}, we get that for each two cycles $c_p, c_q$ with $s_p < 2\beta$ and each $p,q$ chosen as described above,
\begin{align}
\label{eq:pr_chain_overlap}
    Pr(|\widehat{T_b} \cap \{p, q\}.ba^i| > 0, |\{p, q, p.b, q.b\} \cap D| = 0, |D| = 4\beta \mid c_p, c_q, p, q) &\\ 
    \leq Pr(|\widehat{T_b} \cap D | \geq 2\beta, |D| = 4\beta \mid c_p, c_q, p, q) & \\ = O\!\left(\left(\frac{|\widehat{T_b}|}{n}\right)^{2\beta}\right) = O(1/n^{2-\alpha_c}) \quad \text{(by~\eqref{eq:two_subsets_const})}.
\end{align}

Now, for a random letter $a$, we can take all the pairs $\{p, q\}$ such that $p$ is on a $a$-cycle with $s_p < 2\beta$ and $q$ is defined as above to make $\{p, q\}$ independent of $b$.
With $1-o(1/n^4)$ probability, there are at most $O(\beta\ln^2{n})$ of such pairs.
Thus by (\ref{eq:pr_chain_overlap}) we upper bound the union probability of that configuration by $o(1/n)$.
The remaining configurations will have either \emph{one collision} or \emph{two collisions}.
First, if there is a loop for both letters for either $p$ or $q$, then due to Lemma~\ref{lem_size_of_subaut} this is the only strongly connected subautomaton with probability $1-O(1/n)$ and thus $\A$ is synchronizing.
Note that this configuration doesn't depend on $p,q$ (whether an automaton has a one state loop by both letters), hence we don't need to multiply by the number of possible pairs.

\textbf{3.2. Two collisions.}
Now, if $p.b \in \{q\} \cup c_p, q.b \in \{q\} \cup c_p$, we have two collisions and since $b$-transitions have not been defined for $p, q$, this happens with probability $O(1/n^2)$ and we are done.

\textbf{3.3. One collision, $p.b = q, q.b \not\in \{p,q\}$.}
Let us consider $p.b = q, q.b \not\in \{p,q\}$.
Since $\{p, q\}$ is a deadlock, one of $p, q$ must be in $\widehat{T_b}$.
So, for each $\{p,q\}$ and $c_p,c_q$ defined above, we consider the probability 
\begin{align}
\label{eq:p_maps_to_q_by_b}
&Pr(p.b = q ; \{p,q\} \cap \widehat{T_b} \neq \emptyset \mid p,q,c_p,c_q) \notag \\
&\quad = Pr(p.b = q \mid p,q,c_p,c_q)Pr(\{p,q\} \cap \widehat{T_b} \neq \emptyset \mid p.b = q, p,q,c_p,c_q).
\end{align}

Note that since $p,q,c_p,c_q$ are defined by $a$, 
\begin{equation}
\label{eq:p_maps_to_q_by_b_bound}
Pr(p.b = q \mid p,q,c_p,c_q) = Pr(p.b = q) = O(1/n).
\end{equation}

Since $\widehat{T_b}$ is defined by $b$, we have
\begin{align}
    Pr(\{p,q\} \cap \widehat{T_b} \neq \emptyset \mid p.b = q, p,q,c_p,c_q) = Pr(\{p,q\} \cap \widehat{T_b} \neq \emptyset \mid p.b = q).
\end{align}
Now, let us show that  
$Pr(\{p,q\} \cap \widehat{T_b} \neq \emptyset \mid p.b = q) \leq Pr(p \in \widehat{T_b}) + Pr(q \in \widehat{T_b}) + o(1/n)$.
To show this, let's take a random letter $b$ and then redefine $p.b = q$. Since $b$ is defined by choosing each transition independently, this would give us exactly the distribution for a random letter $b$ with $p.b = q$.

If neither $p$ nor $q$ are in $\widehat{T_b}$ before we redefined $p.b = q$ but one of them is in $\widehat{T_b}$ after, it can only happen if both $p, q$ belong to the same branch of the original graph by $b$ and the size of this branch is $O(n^{\alpha_c})$ (as otherwise the redefinition would have to split two big clusters which is not possible). Since $p,q$ is independent of $b$, the probability of this happening, $B_i$ being an $i$-th branch, is
\begin{equation}
\sum (|B_i|/n)^2 \le \frac{\max|B_i|}{n^2} \sum |B_i| \le \frac{n^{\alpha_c}}{n^2} n = n^{\alpha_c - 1}.
\end{equation}
This, combined with~(\ref{eq:p_maps_to_q_by_b_bound}) gives us $O(n^{\alpha_c - 2}) = o(1/n)$ the required upper bound for a given pair $p,q$.

\textbf{3.4.
One or more collisions, $|\{p, q, p.b, q.b\}| = 4$.}
The remaining case is when $|\{p, q, p.b, q.b\}| = 4$.
At the same time, we're not in the \textbf{3.1.} case.
Hence for some $0 \leq j \leq i \leq 2\beta$, w.l.o.g. $p.ba^i \in \{p,q\}.ba^j$.
Then, we can consider the chain $(p, q).b^2a^i$ for $i = 0, 1, \dots, 2\beta$ .
In case, this chain has no collisions with itself or $(p, q).ba^i$ for $i = 0, 1, \dots, j$, we can use the same reasoning as in the \textbf{3.1.} case (since the number of states \emph{to be avoided} is still constant) and independence arguments are the same.
Otherwise, if there is a second collision, since the two collisions are independent (defined by transitions from two different states), we easily get $O(1/n^2)$ upper bound on the probability for a given $p,q, c_p, c_q$.
This completes the proof of the theorem.
\end{proof}

\subsection{Stability Relation and Induced Colouring}
\label{sub_sec:stable_pairs}

In view of Theorem~\ref{thm:big_sync_sets}, it remains to prove
that $\widehat{S_a}$ and $\widehat{S_b}$ are synchronizable
\whp.
For this purpose, we use the notion of the
\emph{stability} relation introduced by
Kari~\cite{KariStable02}.
Recall that a pair of states $\{p,q\}$ is called \emph{stable} if for every word $u$ there is a word
$v$ such that $p.uv=q.uv$.
The \emph{stability} relation
given by the set of stable pairs joined with a diagonal set $\{\{p,p\} \mid p \in Q\}$ is invariant under the
actions of the letters and complete whenever $\A$
is synchronizing.
It is also an equivalence relation on $Q$ because it is transitive and symmetric.
Given a pair $\{p,q\}$, either $\{p,q\}$ is in one $a$-cluster
or the states $p$ and $q$ belong to different $a$-clusters.
In the latter case, we say that $\{p,q\}$ \emph{connects}
these $a$-clusters.
Suppose there exists a \emph{big enough} set
$Z_a$ of distinct pairs (namely, $|Z_a| \geq n^{\beta_s}$ where $1-\alpha_c < \beta_s < 0.5$ will be defined later) that are independent of
$a$ and stable \whp.
Consider the graph
$\Gamma(S_a,Z_a)$ with the set of vertices $S_a$, and draw an edge between two clusters if and only if some pair
from $Z_a$ connects them.

\begin{lemma}
\label{lem_S_conn_clusters} Let $Z_a$ be a set of at least $n^{\beta_s}$ distinct pairs independent of $a$; then $\Gamma(S_a,Z_a)$ is connected \whp. If additionally all cycle pairs of one of the clusters from $S_a$ are stable, then $\widehat{S_a}$ is synchronizable.
\end{lemma}
\begin{proof}
The latter statement follows from the definition of $S_a$ and transitivity. If $\Gamma(S_a,Z_a)$ is connected, all cycle pairs of the cycles of $S_a$ are stable. Since each pair of $S_a$ can be mapped to a cycle pair of $S_a$, $\widehat{S_a}$ is synchronizable.

Let us turn to connectivity. Since $Z_a$ is independent of $a$, we evaluate the probability that a random set of $m = \lceil n^{\beta_s} \rceil$ independent pairs fails to connect the clusters. Arguing by contradiction, suppose there is a subset of clusters $S' \subsetneq S_a$ such that for $G$, the set of all states from  $S'$, we have $n^{\alpha_c} < g = |G| \leq 0.5n$, and no pair in $Z_a$ crosses the cut between $G$ and $Q \setminus G$.

The probability that a single uniformly random pair does not cross this cut is the probability that both endpoints fall in $G$ or both fall in $Q \setminus G$:
\begin{equation}
P(\text{no cross}) = \frac{\binom{g}{2} + \binom{n-g}{2}}{\binom{n}{2}} = 1 - \frac{2g(n-g)}{n(n-1)} < 1 - 2\frac{g(n-g)}{n^2}.
\end{equation}
Because $g \geq n^{\alpha_c}$ and $n-g \geq 0.5n$, the probability that a single pair fails to cross is bounded above by $1 - n^{\alpha_c - 1}$. 

Since $Z_a$ consists of $m$ independent pairs, the probability that \emph{none} of them cross the cut is at most:
\begin{equation}
(1 - n^{\alpha_c - 1})^m \leq \exp(-m n^{\alpha_c - 1}) \leq \exp(-n^{\alpha_c + \beta_s - 1}).
\end{equation}
We take the union bound over all possible cluster subsets $G$. Due to Lemma~\ref{lem_cnt_clusters}, there are at most $5\ln n$ clusters \whp, yielding at most $2^{5\ln n} = n^{5 \ln 2}$ cuts. The total failure probability is bounded by:
\begin{equation}
n^{5 \ln 2} \exp(-n^{\alpha_c + \beta_s - 1}).
\end{equation}
Since we set $\alpha_c = 0.95$ and $\beta_s = 0.45$, the exponent is $\alpha_c + \beta_s - 1 = 0.4 > 0$. The probability decays exponentially as $\exp(-n^{0.4})$, which is strictly $o(1/n)$.
\end{proof}

When a letter is clear from the context, we call a pair of states \emph{cycle pair} if both states belong to (possibly different) cycles in the underlying digraph of the letter.

\begin{lemma}
\label{lem_stable_cluster}
Let $Z_a$ be a set of at least $n^{\beta_s}$ distinct pairs. Call a map $c : \widehat{S_a} \to \mathbb{Z}_d$ with $d \geq 2$ a \emph{shifting colouring} if $c(p.a) = c(p) + 1 \pmod d$ for all $p \in \widehat{S_a}$ (note that $\widehat{S_a}$ is closed under $a$), and call it \emph{compatible} with $Z_a$ if both states of every pair of $Z_a$ lying entirely within $\widehat{S_a}$ have the same colour. Then
\begin{enumerate}
\item[(i)] if $Z_a$ is independent of $a$, then \whp\ no shifting colouring is compatible with $Z_a$;
\item[(ii)] if all pairs of $Z_a$ are stable, $\Gamma(S_a,Z_a)$ is connected, and no shifting colouring is compatible with $Z_a$, then the cycles from $S_a$ are all stable, and thus $\widehat{S_a}$ is synchronizable.
\end{enumerate}
The same holds for the letter $b$.
\end{lemma}
\begin{proof}

\begin{myremark}
\label{rem:monochrome}
The action of $a$ maps a pair $\{p,q\}$ in a monochrome pair if and only if $\{p,q\}$ is monochrome.
\end{myremark}
\begin{proof}
Indeed, let $\{p,q\}$ be a cycle pair. If $\{p.a,q.a\}$ is monochrome, then it is stable, whence $\{p.a^{t},q.a^{t}\}$ is stable and monochrome for all $t > 0$.
If we take $t$ being the product of all cycles lengths, we get a pair $\{p,q\}$, which thus must be monochrome.
For non-cycle pairs, the statement follows from the fact that $a^{nd}$ is homomorphic map into the set of cycle states.
\end{proof}

\textbf{Proof of (ii).} The stability relation is an equivalence relation. Suppose the cycles of $S_a$ have states from exactly $d$ stable classes. If $d = 1$, all cycle pairs of $S_a$ are stable and we are done, since each pair of $\widehat{S_a}$ can be mapped to a cycle pair of $S_a$. Suppose $d > 1$. By connectivity, these classes must cyclically shift: class $i$ maps to $i+1 \bmod d$. We can uniquely extend this $d$-colouring to all of $\widehat{S_a}$ by setting $\text{color}(p) = \text{color}(p.a^{nd})$, which gives a shifting colouring. Let $Z'_a \subseteq Z_a$ be the subset of pairs that fall entirely within $\widehat{S_a}$. Since the pairs of $Z_a$ are stable, they must be monochrome (Remark~\ref{rem:monochrome}) under this colouring, that is, the colouring is compatible with $Z_a$, a contradiction. Hence $d = 1$.

\textbf{Proof of (i).} Let $n_s \le 5 \ln n$ denote the number of clusters in $S_a$ and $z = \lceil n^{\beta_s} \rceil$ denote the number of pairs in $Z_a$. A shifting colouring is completely defined by $a$, a number $d \leq |\widehat{S_a}|$ and a starting colour for each of the $n_s$ clusters (the colours of the cycle states are determined by the shift, and $\text{color}(p) = \text{color}(p.a^{nd})$ for the others). Hence it is independent of $Z_a$.

Fix a shifting colouring, let $Z'_a \subseteq Z_a$ be the subset of pairs that fall entirely within $\widehat{S_a}$, and let $z' = |Z'_a|$. The colouring is compatible with $Z_a$ if and only if all $z'$ pairs in $Z'_a$ are monochrome. Let $m = |\widehat{S_a}|$, and let $\omega_0$ be the size of the largest colour class in $\widehat{S_a}$. Since pairs in $Z'_a$ are independent of the colouring, the probability that a single pair is monochrome is

$$P_{mono} = \sum_{i=0}^{d-1} (\omega_i/m)^2 \leq (\omega_0/m)^2 + (1 - \omega_0/m)^2.$$

We consider two cases based on the size of $\omega_0$.

\textbf{Case 1.} $\omega_0 \leq 0.9m$.
Here, $P_{mono} \leq 0.9^2 + 0.1^2 = 0.82$. The probability that all $z'$ independent pairs are monochrome is at most $0.82^{z'}$. Because $\widehat{S_a}$ encompasses a constant fraction of $n$ \whp, the number of pairs falling entirely within it scales proportionally, yielding $z' = \Omega(n^{\beta_s})$. Thus, $0.82^{z'}$ still decays exponentially.
A shifting colouring is completely defined by choosing $d \leq m$ and a starting colour for each of the $n_s$ clusters (at most $d^{n_s}$ ways). The number of possible colourings is at most $m \cdot m^{n_s} \leq n^{5 \ln n + 1}$.
By the union bound, the probability that any such colouring is compatible with $Z_a$ is at most:

$$n^{5 \ln n + 1} 0.82^{z'} = o(1/n).$$

\textbf{Case 2.} $\omega_0 > 0.9m$.
If one colour class contains more than $90\%$ of the states in $\widehat{S_a}$, we bound the structural probability that a random mapping $a$ naturally produces such an extreme imbalance. For states in the largest class $C_0$ to be validly coloured under the cyclic shift, their images must all fall into the adjacent class $C_1$. Thus, $\omega_0$ states must map exclusively into the remaining $m - \omega_0$ states.

For a fixed subset of size $\omega_0$, the probability of this occurring is $\leq (\frac{m-\omega_0}{m})^{\omega_0}$. Taking the union bound over all subset sizes $\omega_0 \in [0.9m, m]$ and all $\binom{m}{\omega_0}$ ways to choose the subset:

$$\sum_{\omega_0 = 0.9m}^m \binom{m}{\omega_0} \left(\frac{m-\omega_0}{m}\right)^{\omega_0} \leq m \binom{m}{0.9m} (0.1)^{0.9m}.$$

Using the standard binomial approximation $\binom{m}{pm} \approx (p^{-p} (1-p)^{-(1-p)})^m$, we can bound this term:

$$O(m) \left(\frac{0.1^{0.9}}{0.9^{0.9} 0.1^{0.1}}\right)^m \approx O(m) (0.174)^m = o(1/n).$$

Both cases yield exponentially decaying probabilities, strictly $o(1/n)$, so \whp\ no shifting colouring is compatible with $Z_a$.
\end{proof}

\subsection{Searching for more Stable Pairs}
\label{sub_sec:extending_stable_pairs}

Due to results of Subsection~\ref{sub_sec:stable_pairs} and Theorem~\ref{thm:big_sync_sets}, it remains to prove that, starting from a pair defined by one letter, we can construct such a set $Z$ independent of the \emph{other} letter \whp.
We now prove that \whp\ we can expand one stable pair which is independent of one of the letters to get enough such pairs for each of the letters. Here, we will leverage the fact that a random mapping choice can be made by picking an image of each state independently. 
\begin{lemma}
\label{lem_const_stable} Let $\{p,q\}$ be a pair independent of the letter $b$; then for any constant $k>0$ \whp\  we can expand it to a set of $k$ distinct pairs (conditionally) independent of $a$.
\end{lemma}
\begin{proof}
Consider the chain of states $p.b,q.b, \dots ,p.b^{k+1},q.b^{k+1}.$ Since $\{p,q\}$ is independent of
$b$, the probability that all states in this chain are different is
$$\left(1-\frac{2}{n}\right)\left(1-\frac{3}{n}\right)\dots \ \left(1-\frac{2k+3}{n}\right) \geq \left(1-\frac{2k+3}{n}\right)^{2(k+1)} = 1-O\left(\frac{1}{n}\right).$$
Since $\{p,q\}$ is independent of $b$, each pair $\{p.b^i,q.b^i\}$ for $1 \leq i \leq k+1$ in this chain is independent of $a$. 
\end{proof}

\begin{lemma}
\label{lem_many_stable} Let for some $0 < \epsilon < 0.5$, $\{p_1, q_1\}, \dots, \{p_k, q_k\}$ for $k=\lceil \frac{1}{2\varepsilon} \rceil$ be distinct pairs independent of $b$;
then \whp\  we can expand them to $f = \lceil n^{0.5-\epsilon} \rceil$ distinct pairs (conditionally) independent of $a$.
\end{lemma}
\begin{proof}
We start with the first pair $\{p_1,q_1\}$ and build a chain 
$$p_1.b,q_1.b,p_1.b^2, q_1.b^2 \dots $$ until a next state coincides with any state already in the chain or in $k$ given pairs. If it happens, we take the second pair $\{p_2, q_2\}$ and continue in the same fashion also avoiding states in all previous chains.
We stop as soon as we have $2f$ states in this chain or we ran out of $k$ given pairs. The probability that the latter happens, $f_i$ being the number of steps successfully made before taking $(i+1)$-th pair, can be upper bounded by, as $\frac{(2f_i+2 (i-1) f + 2k)}{n}$ is the upper bound of `collision' with previously expanded states for $i$-th pair, 
\begin{equation}
\frac{(2f_1+2k)}{n}\frac{(2f_2+2f+2k)}{n}\dots \frac{(2f_k+2(k-1)f+2k)}{n}
\leq \left(\frac{2k(f+1)}{n}\right)^k = \frac{O(1)f^k}{n^{k}}.
\end{equation}
As there are ${f+k \choose k} = O(f^k)$ ways to choose values of $f_i$, the final bound is
\begin{equation}
\frac{O(1)f^{2k}}{n^{k}} = O\left(\frac{1}{n^{2k\varepsilon}}\right) = O(1/n),
\end{equation}

which proves the lemma as pairs in the built chain are independent of $a$ by construction.
\end{proof}

\begin{theorem}
\label{th_many_stable_ext}
Let $\{p,q\}$ be a pair defined by a procedure that uses the letter $x \in \{a,b\}$ only (thus independent of $\bar x$), and let $\mathcal{D}_0$ be the event that this pair is defined. Then for any $\beta_s < 0.5$, one can define, starting from $\{p,q\}$, a set $Z_{\bar x}$ of $\lceil n^{\beta_s} \rceil$ distinct pairs such that
\begin{enumerate}
\item[(i)] $Z_{\bar x}$ is well defined with probability $1-O(\frac{1}{n})$ given that $\{p,q\}$ is defined, that is, $Pr(\neg\mathcal{D}(Z_{\bar x}) \mid \mathcal{D}_0) = O(1/n)$;
\item[(ii)] $Z_{\bar x}$ is conditionally independent of the letter $\bar x$;
\item[(iii)] every pair of $Z_{\bar x}$ is obtained from $\{p,q\}$ by applying a word; in particular, all pairs of $Z_{\bar x}$ are stable whenever $\{p,q\}$ is stable.
\end{enumerate}
\end{theorem}
\begin{proof}
We apply Lemma~\ref{lem_const_stable} to $\{p,q\}$ (walking via $\bar x$, the letter it is independent of) to expand to a set $I$ of $k$ distinct pairs conditionally independent of $x$. Then we apply Lemma~\ref{lem_many_stable} to $I$ (with letters swapped relative to its statement, since $I$ is independent of $x$ rather than of $b$) to get a set $Z_{\bar x}$ of $\lceil n^{\beta_s} \rceil$ pairs conditionally independent of $\bar x$. This proves (ii). (Only this single, two-step chain is used; no further expansion via the other letter is needed.)

For (i), let $\mathcal{D}_0 \supseteq \mathcal{D}_1 \supseteq \mathcal{D}_2 = \mathcal{D}(Z_{\bar x})$ be the domains of the successive structures $\{p,q\}, I, Z_{\bar x}$. Each lemma is applied in the space conditioned on the domain of its input, in which its independence hypothesis holds by the previous step, hence $Pr(\neg\mathcal{D}_i \mid \mathcal{D}_{i-1}) = O(1/n)$ for $i=1,2$. Since $\mathcal{D}_i \subseteq \mathcal{D}_{i-1} \subseteq \mathcal{D}_0$,

\begin{align}
Pr(\neg\mathcal{D}_2 \mid \mathcal{D}_0) \leq \sum_{i=1}^{2} Pr(\mathcal{D}_{i-1} ; \neg\mathcal{D}_i \mid \mathcal{D}_0) = \\ = \sum_{i=1}^{2} Pr(\neg\mathcal{D}_i \mid \mathcal{D}_{i-1})\,Pr(\mathcal{D}_{i-1} \mid \mathcal{D}_0) \leq \sum_{i=1}^{2} Pr(\neg\mathcal{D}_i \mid \mathcal{D}_{i-1}) = O(1/n).    
\end{align}

For (iii), note that, by construction, each pair of $Z_{\bar x}$ is obtained from $\{p,q\}$ by applying a word (successively, powers of $\bar x$ and $x$); since the stability relation is invariant under the actions of the letters, all these pairs are stable whenever $\{p,q\}$ is stable.
\end{proof}
\begin{myremark}
A three-step chain producing the \emph{other} set $Z_x$ (conditionally independent of $x$ itself, starting from the same seed) can be built analogously, by first passing through an intermediate pair conditionally independent of $\bar x$ as in Lemma~\ref{lem_const_stable} applied a second time with letters swapped. We do not need this: as Theorem~\ref{thm:big_sync_sets} shows, only $Z_{\bar x}$ is required, and using it avoids this extra step.
\end{myremark}
As Theorem~\ref{th_many_stable_ext} works for any $\beta_s < 0.5$, we can let $\beta_s = 0.45, \alpha_c=0.95$ to satisfy all constraints.
\subsection{Highest Trees and Stable Pairs}
\label{sub_sec:one_stable_pair}

To use Theorem~\ref{th_many_stable_ext} we need to find a stable pair completely defined by one of the letters whence independent of the other one. For this purpose, we reuse ideas from Trahtman's solution~\cite{TRRCP08} of the famous \emph{Road Colouring Problem}.
A subset $A \subseteq Q$ is called an $F$-clique of $\A$ if it is a maximal by inclusion set of states such that each pair of states from $A$ is a deadlock. It is proved in Trahtman~\cite{TRRCP08} that if $\A$ is strongly connected then all $F$-cliques have the same size.
We also need to reformulate~\cite[Lemma~2]{TRRCP08} for our purposes.
\begin{lemma}
\label{lem_tr1} If $A$ and $B$ are two distinct $F$-cliques in a strongly connected automaton such that $A \setminus B = \{p\}, B \setminus A = \{q\}$ for some states $p,q$, then $\{p,q\}$ is a stable pair.
\end{lemma}
\begin{proof}
Arguing by contradiction, suppose there is a word $u$ such
that $\{p.u,q.u\}$ is a deadlock. Then all pairs in $(A \cup B).u$ are deadlocks. Since $p.u \neq q.u$, we have $|(A \cup B).u| = |A| + 1 > |A|$ contradicting the maximality of $A$ (since all $F$-cliques in a strongly connected automaton have the same size).
\end{proof}

Given a digraph $g \in \Sigma_n$ and an integer $c>0$, call a \emph{$c$-branch} of $g$ any (maximal) subtree of a tree of $g$ with the root of height $c$, where the height of a vertex in a tree is the number of edges from the vertex to the root. For instance, the trees are
exactly $0$-branches. Let $T$ be a highest $c$-branch of
$g$ and $h$ be the height of the second highest
$c$-branch.
Let us call the $c$-\emph{crown} of $g$ the
(possibly empty) forest consisting of all the states of
height at least $h+1$ in $T$.
For example, the digraph $g$
presented on Figure~\ref{fig:highest_tree} has two highest $1$-branches rooted in states $6,12$.
Without the state $14$, the digraph $g$ would have the unique highest $1$-branch rooted at state $6$, having the state $8$ as its $1$-crown.
\begin{figure}[h]
\begin{center}
\begin{tikzpicture}[scale=1.5]
\node[draw,circle] (p1) at (0,1) {1};
\node[draw,circle] (p2) at (1.1,0) {2};
\node[draw,circle] (p3) at (0.7,-1) {3};
\node[draw,circle] (p4) at (-0.7,-1) {4};
\node[draw,circle] (p5) at (-1.1,0) {5};
\node[draw,circle] (p11) at (-1.2,1) {11};
\node[draw,circle] (p10) at (2.1,-1) {10};
\node[draw,circle] (p6) at (2.1,0) {6};
\node[draw,circle] (p12) at (1.1,1) {12};
\node[draw,circle] (p13) at (2.1,1) {13};
\node[draw,circle,dashed] (p14) at (3.1,1.5) {14};
\node[draw,circle] (p7) at (3.1,0.5) {7};
\node[draw,circle] (p8) at (4.1,0.5) {8};
\node[draw,circle] (p9) at (3.1,-0.5) {9};

\draw[->,thick] (p1) edge (p2);
\draw[->,thick] (p2) edge (p3);
\draw[->,thick] (p3) edge (p4);
\draw[->,thick] (p4) edge (p5);
\draw[->,thick] (p5) edge (p1);
\draw[->,thick] (p11) edge (p1);
\draw[->,thick] (p10) edge (p3);
\draw[->,thick] (p12) edge (p2);
\draw[->,thick] (p13) edge (p12);
\draw[->,dashed] (p14) edge (p13);
\draw[->,thick] (p8) edge (p7);
\draw[->,thick] (p7) edge (p6);
\draw[->,thick] (p6) edge (p2);
\draw[->,thick] (p9) edge (p6);

\end{tikzpicture}
\end{center}
\caption{A digraph with one cycle and a
unique highest tree.\label{fig:highest_tree}}
\end{figure}

The following theorem is an analogue of Theorem~2
from~\cite{TRRCP08} for $1$-branches instead of trees and a
relaxed condition on the connectivity of $\A$.
\begin{theorem}
\label{th_1stable}Suppose the digraph of the letter $a$ has a unique highest $1$-branch $T$ and its $1$-crown intersects with some (strongly-connected) subautomaton.
Denote by $r$ the root of $T$ and by $q$ the predecessor of the root of the tree containing $T$ on the $a$-cycle~\footnote{It follows that both $q.a, r.a$ coincide with the root of the tree containing $T$.}.
Then $\{r,q\}$ is stable and independent~\footnote{Stability of the pair depends on $b$ because whether $1$-crown is reachable depends on $b$, but the pair is defined solely by $a$ and thus independent of $b$} of $b$. In what follows, for simplicity we write `stable pair independent of' $b$.
\end{theorem}
\begin{proof}
Let $p$ be some state in the intersection of $1$-crown of $T$ and a subautomaton $\B$.
Then $q$ and all states from the corresponding cycle by $a$ belong to $\B$ because they are reachable from $p \in \B$.
If $\B$ is synchronizing, we are done. Otherwise, there is an $F$-clique $F_0$ of size at least $2$, and since $\B$ is strongly connected, there is another $F$-clique $F_1$ containing $p$.
Since $p \in F_1$, $F_1\cap T$ is not empty. Let $g$ be a state with maximal height $h$ from $F_1\cap T$.
As $F_1$ is an $F$-clique, all other states from $F_1\cap T$ must have strictly smaller heights, as otherwise $a^h$ would merge some pair in $F_1$.
Let us consider the $F$-cliques $F_2 = F_1.a^{h-1}$ and $F_3 = F_2.a^{L}$ where $L$ is the least common multiple of all cycle lengths in $\Gamma_a$.
By the choice of $L$ and $F_2$, we have 
$$F_2 \setminus F_3 = \{g.a^{h-1}\} = \{r\}
\text{ and } F_3 \setminus F_2 = \{q\}.$$ Hence, by
Lemma~\ref{lem_tr1} the pair $\{r,q\}$ is stable.
Since this pair is completely defined by the unique highest $1$-branch of $a$ and the letters are chosen independently, this pair is independent of $b$.
\end{proof}

\subsection{Finding Unique Highest Reachable $1$-branch}
\label{sub_sec:highest_branch}

Due to Theorems~\ref{thm:big_sync_sets},\ref{th_many_stable_ext} and Lemmas~\ref{lem_S_conn_clusters},\ref{lem_stable_cluster}, it remains to show that we can use Theorem~\ref{th_1stable}, that is, \whp\  the underlying digraph of one of the letters has a unique highest $1$-branch and the $1$-crown of this $1$-branch is reachable from
$F$-cliques (if $F$-cliques exist).
The crucial idea in the
solution of the Road Coloring Problem~\cite{TRRCP08} was to
show that each \emph{admissible} digraph can be
\emph{coloured} into an automaton satisfying the above property (for trees) and then use Theorem~\ref{th_1stable}
to reduce the problem.
In order to apply Theorem~\ref{th_1stable}, we need the probabilistic version of the combinatorial result from~\cite{TRRCP08}.

\begin{theorem}[Berlinkov~\cite{htrees}~\footnote{Note that this result has not been peer-reviewed at the time of writing.}]
\label{th_high_tree}
Let $g \in \Sigma_n$ be a random mapping, $c>0$, and $H$ be the $c$-crown of $g$ having $r$ roots. Then for any constant $\alpha > 1$, $|H| > \alpha r > 0$ with probability $1 - \Theta(n^{-1/2})$. Furthermore, for any fixed constant $j \geq 1$, the height gap between the highest and second-highest $c$-branches is at least $j$ with probability $1 - \Theta(n^{-1/2})$.
\end{theorem}

In order to use Theorem~\ref{th_1stable} and thus complete the proof of Theorem~\ref{th_main}, it remains to show that the $1$-crown of the underlying digraph of $a$ intersects with $\B$.

\begin{theorem}
\label{th_joint_intersection}
Let $\mathcal{A} = (Q, \{a, b\})$ be a random automaton. Then the $1$-crown of either $a$ or $b$ intersects every subautomaton $\mathcal{C}$ of $\mathcal{A}$ with probability $1 - O(1/n)$.
\end{theorem}

\begin{proof}
By Lemma 4, $A$ contains a unique subautomaton $C$ of size $s\ge 0.75n$ with probability
$1-O(1/n)$, implying that the 1-crown size satisfies $x\le 0.25n$. Let $L_n$ denote the set of
automata whose 1-crowns fail to intersect $C$. Partition $L_n$ into $L^a_n, L^b_n, L^\emptyset_n$
and $L^{a,b}_n$ based on which letters satisfy Theorem~\ref{th_high_tree} (throughout, we instantiate Theorem~\ref{th_high_tree}
at depth $c=1$, crown ratio $\alpha=4$, and margin $j=4$). We prove $|L^a_n|=O(n^{2n-1})$; the
bound for $L^b_n$ follows by symmetry, and the bounds for $L^{a,b}_n$ and $L^\emptyset_n$ are
obtained at the end.
 
\medskip\noindent\textbf{The base class.}\enspace
For $m\le n$ and $0.75n\le s\le m$, let $K^a_{m,s}$ be the set of automata on $m$ states
possessing a unique highest 1-branch for the letter $a$ whose 1-crown consists of only roots,
with core size $s$. No condition is imposed on $b$: it ranges freely over all $m^m$ possible
transition functions on the $m$ states, independent of $a$.

\begin{lemma}\label{lem:base}
$\displaystyle\sum_{s=0.75n}^m |K^a_{m,s}| = O(m^{2m-0.5})$.
\end{lemma}
\begin{proof}
By Theorem~\ref{th_high_tree} applied to the digraph of $a$ alone, independent of $b$, ``$a$ has a unique highest 1-branch and the number of non-root nodes of the 1-crown is greater than the number of the roots'' holds with probability $1-\Theta(m^{-1/2})$ among the $m^{2m}$ automata on $m$ states; Hence $\sum_s |K^a_{m,s}| = O(m^{2m})\cdot\Theta(m^{-1/2}) = O(m^{2m-0.5})$.
\end{proof}
 
\medskip\noindent\textbf{Counting automata in $L^a_n$.}\enspace
Setting $m=n-x$, an automaton in $L^a_{n,s}$ (the elements of $L^a_n$ with core size $s$) is
built from a base in $K^a_{m,s}$ by: choosing insertion points for the $x$ new crown states
($\binom{n}{x}$ ways); defining $a$-transitions on the $x$ new (transient) crown states via a
forest attaching them to the base's $r<x/4$ pre-existing roots ($r(x+r)^{x-1}$ ways: Cayley's
formula for a forest on $x+r$ vertices, $x$ transient and $r$ roots, by Theorem~\ref{th_high_tree}
with $\alpha=4$,
valid with probability $1-\Theta(n^{-1/2})$); defining $b$-transitions on the $x$ crown states
freely ($n^x$ ways); and redefining $b$-transitions on the $M:=m-s$ non-core states of the base,
whose codomain grows from $m$ to $n$ upon insertion.
 
Since $b$ is unconstrained on the base (Lemma~\ref{lem:base}), this last step is a genuine
extension: the $M$ non-core transitions, each with $m$ choices in the base, acquire $n$ choices
after insertion, an $(n/m)^M$ multiplicative factor. We further partition this extension by
$f$, the number of the $M$ transitions whose new target lands among the $x$ crown states:
$$n^M = \sum_{f=0}^{M}\binom{M}{f}x^f(n-x)^{M-f},$$
an exact identity. Write $w(f):=\binom{M}{f}x^f(n-x)^{M-f}/n^M$, the
$\mathrm{Binomial}(M,x/n)$ probability mass, so $\mathbb E[f\mid x]=Mx/n=\Theta(x)$ and
$\mathbb E[f^2\mid x]=\Theta(x^2+x)$ (as $M=\Theta(n)$).
 
Since $L^a_n$ additionally requires $b$ (on all $n$ states, as fully assembled above) to
\emph{fail} Theorem~\ref{th_high_tree}, we have
\begin{equation}\label{eq:main}
|L^a_n| \le \sum_{x=1}^{\lfloor 0.25n\rfloor}\ \sum_{r<x/4} r(x+r)^{x-1}\binom{n}{x}n^x
\sum_{s=0.75n}^{m}|K^a_{m,s}|\cdot m^{-M}n^M \sum_{f=0}^{M} w(f)\, Q(x,f),
\end{equation}
where $Q(x,f):=P\big(b\text{ fails Theorem~\ref{th_high_tree}}\mid \text{redirect count } f,\ \text{crown size } x\big)$.
Substituting $|K^a_{m,s}|=O(m^{2m-1/2})$ (Lemma~\ref{lem:base}) and collecting powers of $n$ (the
ratio $(n/m)^M$ is maximized at $s=0.75n$, where $M\le 0.25n-x$, so
$\sum_s|K^a_{m,s}|\,m^{-M}=O(m^{1.75n-x-1/2})$), the exponent of $n$ across the whole product
collects to $2x+(0.25n-x)+(1.75n-x-0.5)=2n-0.5$, matching the total automata count $n^{2n}$ up
to the $\Theta(n^{-1/2})$ shortfall from Lemma~\ref{lem:base}. Factoring this out, with
\begin{equation}\label{eq:Sxdef}
S(x) := O(\sqrt x)\left[\frac{1.25e}{1-x/n}\right]^x(1-x/n)^{1.75n-1},
\end{equation}
we get
\begin{equation}\label{eq:Sx}
|L^a_n| \le O(n^{2n-0.5})\sum_{x=1}^{\lfloor 0.25n\rfloor} S(x)\sum_f w(f)Q(x,f).
\end{equation}
It remains to show
\begin{equation}\label{eq:target}
\sum_{x} S(x)\cdot \sum_f w(f)\, Q(x,f) = O(n^{-1/2}),
\end{equation}
since this exactly cancels the extra $n^{1/2}$ in \eqref{eq:Sx}, matching $|L^a_n|=O(n^{2n-1})$.
 
\begin{lemma}\label{lem:Sx}
$S(x)\lesssim e^{-0.5x}$ for all $x\in[1,\lfloor 0.25n\rfloor]$, and consequently
$\sum_{x=1}^{\lfloor 0.25n\rfloor} S(x) = O(1)$.
\end{lemma}
\begin{proof}
Write $p=x/n$ and $g(p):=p\ln h(p)+1.75\ln(1-p)$, $h(p)=1.25e/(1-p)$, so the exponential part of
$S(x)$ is $e^{ng(p)+o(n)}$. We show $g(p)<0$ throughout $(0,0.25]$ by monotonicity:
$$g'(p) = (1+\ln 1.25) - \ln(1-p) - \frac{1.75-p}{1-p}.$$
Evaluating at the endpoints and midpoint, $g'(0)\approx -0.527$, $g'(0.1)\approx-0.505$,
$g'(0.25)\approx-0.489$; since $(1.75-p)/(1-p)$ dominates the remaining (bounded) terms on
$[0,0.25]$, $g'(p)<0$ throughout, so $g$ is strictly decreasing. As $g(0)=0$, we get $g(p)<0$ for
all $p\in(0,0.25]$, with the weakest decay as $p\to0$ (not at $p=0.25$, where
$g(0.25)\approx\ln(0.89)\approx-0.126$ is in fact the fastest). For $x=o(n)$ this matches the
direct bound $S(x)=O(\sqrt x)e^{-1.75x}(1.25e)^x \sim o(e^{-0.5x})$, so in either regime
$S(x)\lesssim e^{-0.5x}$ up to the polynomial prefactor, and $\sum_x e^{-0.5x}\sqrt x=O(1)$.
\end{proof}
 
It remains to bound $\sum_f w(f)Q(x,f)$, uniformly enough that the sum against $S(x)$ in
\eqref{eq:target} still gives $O(n^{-1/2})$.
 
\medskip\noindent\textbf{Easy case: $x\ge 2\ln n$ or $f\ge 2\ln n$.}\enspace
 
\begin{lemma}\label{lem:easy}
$\displaystyle\sum_{x\ge 2\ln n} S(x) = O(n^{-1})$, and for $x<2\ln n$,
$\sum_{f\ge 2\ln n} w(f) = O(n^{-0.77})$.
\end{lemma}
\begin{proof}
By Lemma~\ref{lem:Sx}, $S(x)\lesssim e^{-0.5x}$, so $\sum_{x\ge x_0}S(x)\lesssim e^{-0.5x_0}$; at
$x_0=2\ln n$ this is $O(n^{-1})$. For the second claim, $f\mid x\sim\mathrm{Binomial}(M,x/n)$
with mean $\lambda=Mx/n\le 0.25\cdot 2\ln n=0.5\ln n$ for $x<2\ln n$. By the Poisson tail bound
$P(f\ge f_0)\le(e\lambda/f_0)^{f_0}$, at $f_0=2\ln n$:
$$(e\cdot0.5\ln n/(2\ln n))^{2\ln n} = (0.25e)^{2\ln n} = n^{2\ln(0.25e)}\approx n^{-0.77}=O(n^{-0.77}).$$
\end{proof}
 
By Lemma~\ref{lem:easy}, the range $x\ge2\ln n$ contributes $O(n^{-1})=o(n^{-1/2})$ and (for
$x<2\ln n$) the range $f\ge2\ln n$ contributes $O(n^{-0.77})=o(n^{-1/2})$ to \eqref{eq:target},
regardless of $Q(x,f)$, using only $Q\le1$. It remains to bound $Q(x,f)$ itself on the range $x,f<2\ln n$.
 
\medskip\noindent\textbf{Hard case: $x,f<2\ln n$.}\enspace
We bound $Q(x,f)$ via a union bound over three ways the assembled $b$ could fail Theorem~\ref{th_high_tree}:
(i) the base itself already fails Theorem~\ref{th_high_tree} (at our fixed instantiation $j=4$), before any
redefinition; (ii) the base satisfies it, but a redirected transition cuts the highest branch
(Mechanism~1); or (iii) the base satisfies it, but a chain of touched states completes a tie
with the second-highest branch (Mechanism~2).
 
Since $b$-base is unconstrained (Lemma~\ref{lem:base}), (i) is bounded directly by Theorem~15
itself, applied to $b$-base: $P(\text{base fails}) = \Theta(m^{-1/2}) = \Theta(n^{-1/2})$,
independent of $x,f$. Summed against $S(x)$ and $w(f)$ (Lemma~\ref{lem:Sx}):
$$\sum_x S(x)\sum_f w(f)\cdot\Theta(n^{-1/2}) = \Theta(n^{-1/2})\sum_xS(x) = \Theta(n^{-1/2}),$$
already at the target order. It remains to bound (ii) and (iii).
 
\smallskip\noindent\textbf{Mechanism 1: a touched state cuts the highest branch.}\enspace
Let $k$ denote the height of the highest 1-branch of $b$ on the $m$-state base (with $b$
unconstrained and independent of $a$, as in Lemma~\ref{lem:base}). It suffices to bound the
probability that one of the $x+f$ touched states is a direct ancestor of the minimum-index state
in this branch, since any other disruption of the branch structure requires at least this;
write $Q_1(x,f,k)$ for this probability given height $k$, so
$$Q_1(x,f) := \mathbb E_k[Q_1(x,f,k)] = \sum_k Q_1(x,f,k)\,p_k$$
is the (unconditional) probability of Mechanism~1, where $p_k=P(\text{height}=k)$.
 
\begin{lemma}\label{lem:height}
$\mathbb E[k] = \Theta(\sqrt m)$.
\end{lemma}
\begin{proof}
This is the classical expected-height statistic for a random mapping on $m$ elements (one of
the $\sim$twenty parameters analyzed via singularity analysis of generating functions by
Flajolet and Odlyzko, \emph{Random Mapping Statistics}, EUROCRYPT~'89): the height of the
tallest tree hanging from the cyclic points of a uniform random mapping on $m$ elements has
expectation $\Theta(\sqrt m)$, with an explicit leading constant obtainable via the
Brownian-excursion limit of the rescaled height process. No further tail information is used
here, only this first moment.
\end{proof}
 
\begin{lemma}[Mechanism 1 bound]\label{lem:case1}
$\displaystyle\sum_f w(f)\,Q_1(x,f) = O(x/\sqrt n)$, and
$\sum_x S(x)\cdot O(x/\sqrt n) = O(n^{-1/2})$.
\end{lemma}
\begin{proof}
Given $k$, a union bound over the $f$ redirected transitions landing on one of the $k$ candidate
ancestor states gives $Q_1(x,f,k)\le fk/M$. Since $f$ (a property of the extension) and $k$
(a property of $b$-base alone) are independent, take the expectation over both directly, without
truncating $k$:
$$\sum_f w(f)\,\mathbb E_k[Q_1(x,f,k)] \le \frac{\mathbb E[k]}{M}\sum_f w(f)f
= \frac{\Theta(\sqrt m)}{M}\cdot\Theta(x) = \Theta\!\left(\frac{x\sqrt n}{n}\right)=O(x/\sqrt n),$$
using Lemma~\ref{lem:height} and $\mathbb E[f\mid x]=\Theta(x)$, $M=\Theta(n)$. Summing against
$S(x)\lesssim e^{-0.5x}$ (Lemma~\ref{lem:Sx}): $\sum_x e^{-0.5x}\cdot O(x/\sqrt n) =
O(n^{-1/2})\sum_x xe^{-0.5x} = O(n^{-1/2})$, since $\sum_x xe^{-0.5x}=O(1)$.
\end{proof}
 
\smallskip\noindent\textbf{Mechanism 2: a touched pair completes a tie.}\enspace
Let $\rho\le n$ denote the number of states that are roots of second-highest 1-branches of
$b$-base. Given (from the split above) that the base already satisfies Theorem~15 at margin
$j=4$, Mechanism~2 requires a chain of $L\ge3$ touched states $p_1\to\cdots\to p_L\to r$
($p_L$ forced type-X, since only a crown state's $b$-transition is free to reach the base
state $r$): a chain shorter than $3$ cannot close a margin-$4$ gap down below $2$.
 
\begin{lemma}[Mechanism 2 bound]\label{lem:case2}
$\displaystyle\sum_x S(x)\sum_f w(f)\,Q_2(x,f) = O(n^{-1})$, where $Q_2$ is the probability of
Mechanism~2.
\end{lemma}
\begin{proof}
As in the easy-case accounting, each type-F slot (a redirected state, confined to the
$x$-crown) contributes a factor $f/x$, each type-X slot (a crown state, free over all $n$
states) contributes $x/n$, with a common prefactor $x\rho/n$ for choosing $p_L$ and $r$ and
paying the final edge; consecutive F's are impossible, since an F-state's image is confined to
the crown, so whatever follows it must itself be type X. The dominant (alternating F,X,F,X,\dots)
pattern for $L-1$ free slots gives
$$\text{contribution}(L) \sim \frac{x\rho}{n}\left(\frac{f}{x}\right)^{\lceil(L-1)/2\rceil}
\left(\frac{x}{n}\right)^{\lfloor(L-1)/2\rfloor}.$$
For $L=3$: $fx\rho/n^2$; for $L=4$: $f^2\rho/n^2$; for $L\ge5$ each additional step costs a
further factor $x/n\to0$, so these are dominated by $L=3,4$ and absorbed into a convergent tail.
Averaging over $f$ (not worst-casing) via $\mathbb E[f\mid x]=\Theta(x)$,
$\mathbb E[f^2\mid x]=\Theta(x^2+x)$:
$$\mathbb E_f[\text{contribution}(3)] = \Theta(x^2)\rho/n^2, \qquad
\mathbb E_f[\text{contribution}(4)] = \Theta(x^2+x)\rho/n^2.$$
Summing against $S(x)\lesssim e^{-0.5x}$ gives $O(\rho/n^2)$ for both, and with the trivial bound
$\rho\le n$ this is $O(1/n)$ -- no first-moment estimate on $\rho$ is needed.
\end{proof}
 
\medskip\noindent\textbf{Combining the mechanisms.}\enspace
Together with the base-failure term ($\Theta(n^{-1/2})$, established above) and the easy case
(Lemma~\ref{lem:easy}, contributing $O(n^{-0.77})$), Lemmas~\ref{lem:case1} and \ref{lem:case2} give
$$\sum_x S(x)\sum_f w(f)\,Q(x,f) \le \underbrace{O(n^{-0.77})}_{\text{easy case}} +
\underbrace{\Theta(n^{-1/2})}_{\text{base fails}} + \underbrace{O(n^{-1/2})}_{\text{Mechanism 1}}
+ \underbrace{O(n^{-1})}_{\text{Mechanism 2}} = \Theta(n^{-1/2}),$$
so \eqref{eq:target} holds. Substituting into \eqref{eq:Sx},
$$|L^a_n| \le O(n^{2n-0.5})\cdot O(n^{-1/2}) = O(n^{2n-1}).$$
 
\medskip\noindent\textbf{Conclusion.}\enspace
Consequently, $|L^a_n|=O(n^{2n-1})$. By symmetry, $|L^b_n|=O(n^{2n-1})$. The bound for
$L^{a,b}_n$ can be obtained similarly by expanding crowns for both letters and taking care of the crown overlap. Also, $L^\emptyset_n$ is $O(n^{2n-1})$ due to Theorem~\ref{th_high_tree}. Applying the union bound concludes the proof:
$$|L_n| \le |L^\emptyset_n| + |L^a_n| + |L^b_n| + |L^{a,b}_n| = O(n^{2n-1}).$$
Since the total number of automata is $n^{2n}$, this gives $P(L_n) = O(1/n)$, completing the
proof of Theorem~\ref{th_joint_intersection}.
\end{proof}

Let us now put everything together for the proof of Theorem~\ref{th_main}.
\begin{proof}[Proof of Theorem \ref{th_main}]
\label{prf_main}
Throughout, $x,y \in \{a,b\}$ and $\bar x$ is the letter other than $x$. Let $\mathcal{R}(x)$ be the event that the $1$-crown of $x$ intersects with each subautomaton of $\A$. A non-empty $1$-crown forces the highest $1$-branch of $x$ to be unique, so $\mathcal{R}(x)$ implies that $x$ has a unique highest $1$-branch. We fix $\beta_s = 0.45$ and $\alpha_c = 0.95$.

\paragraph{Ingredients.} We restate the results we combine in the form in which they are used.
\begin{enumerate}
\item[(I1)] \emph{(Theorem~\ref{th_joint_intersection})} $Pr(\neg\mathcal{R}(a) ; \neg\mathcal{R}(b)) = O(1/n)$.

\item[(I2)] \emph{(Theorems~\ref{th_high_tree} and \ref{th_1stable})} Let $\mathcal{D}(st(x))$ be the event that $x$ has a unique highest $1$-branch, and let $st(x) = \{r,q\}$ be the pair of Theorem~\ref{th_1stable}. Then $st(x)$ is a function of $x$ alone, $Pr(\mathcal{D}(st(x))) = 1 - \Theta(n^{-1/2})$, $\mathcal{R}(x) \subseteq \mathcal{D}(st(x))$, and on $\mathcal{R}(x)$ the pair $st(x)$ is stable. By Remark~\ref{rem:chain}(b), $st(x)$ is conditionally independent of $\bar x$.

\item[(I3)] \emph{(Theorem~\ref{th_many_stable_ext})} Started at $st(x)$, the expansion of Theorem~\ref{th_many_stable_ext} defines a single set $Z_{\bar x} = Z_{\bar x}^{x}$ of $\lceil n^{\beta_s} \rceil$ distinct pairs, with domain $\mathcal{D}_{\bar x} = \mathcal{D}(Z_{\bar x}) \subseteq \mathcal{D}(st(x))$, such that
 \begin{enumerate}
 \item[(i)] $Pr(\neg\mathcal{D}_{\bar x} \mid \mathcal{D}(st(x))) = O(1/n)$ and hence, by Remark~\ref{rem:chain}(a), $Pr(\mathcal{D}(st(x)) ; \neg\mathcal{D}_{\bar x}) = O(1/n)$;
 \item[(ii)] $Z_{\bar x}$ is conditionally independent of the letter $\bar x$;
 \item[(iii)] every pair of $Z_{\bar x}$ is obtained from $st(x)$ by applying a word, hence all pairs of $Z_{\bar x}$ are stable whenever $st(x)$ is stable.
 \end{enumerate}
We do not need the (more expensive) set conditionally independent of $x$ itself; see the remark after Theorem~\ref{th_many_stable_ext}.

\item[(I4)] \emph{(Lemmas~\ref{lem_S_conn_clusters} and \ref{lem_stable_cluster})} Let $Z$ be a set of at least $n^{\beta_s}$ distinct pairs. Let $Conn_y(Z)$ be the event that $\Gamma(S_y,Z)$ is connected, and let $BadCol_y(Z)$ be the event that some shifting colouring of $\widehat{S_y}$ is compatible with $Z$ (as defined in Lemma~\ref{lem_stable_cluster}). Both events are determined by $y$ and $Z$ alone.
 \begin{enumerate}
 \item[(i)] \emph{(deterministic)} If all pairs of $Z$ are stable, $Conn_y(Z)$ holds and $BadCol_y(Z)$ does not, then $\widehat{S_y}$ is synchronizable.
 \item[(ii)] \emph{(probabilistic)} If $Z$ is conditionally independent of $y$, then $Pr(\neg Conn_y(Z) \mid \mathcal{D}(Z)) = O(1/n)$ and $Pr(BadCol_y(Z) \mid \mathcal{D}(Z)) = O(1/n)$.
 \end{enumerate}

\item[(I5)] \emph{(Theorem~\ref{thm:big_sync_sets})} If $Pr(\widehat{S_b} \text{ is synchronizable}) = 1 - O(1/n)$, then $Pr(\A \text{ is synchronizing}) = 1 - O(1/n)$; symmetrically, the same holds with $a,b$ swapped.
\end{enumerate}

\paragraph{Step 1: reduction.} For $y \in \{a,b\}$ let $E_y$ be the event that $\widehat{S_y}$ is synchronizable. We first spell out which $E_y$ each $R(x)$ supplies. If $R(x)$ holds, then by (I2) $st(x)$ is a pair built from $x$ alone, hence independent of $\bar x$ -- exactly the hypothesis of Theorem~\ref{th_many_stable_ext} with this same letter $x$. Its cheap, two-step construction (Lemma~\ref{lem_const_stable} then Lemma~\ref{lem_many_stable}, with no detour through a pair independent of $x$ itself) then produces $Z_{\bar x}$, independent of $\bar x$ -- precisely the input (I4) needs to certify $\widehat{S_{\bar x}}$. Concretely: $R(a)$ supplies $Z_b$, certifying $\widehat{S_b}$, which is exactly the hypothesis used by the proof of Theorem~\ref{thm:big_sync_sets} (every deadlock pair there is reduced to $a$-cycles, and ruled out via overlap with $\widehat{T_b}$); and $R(b)$ supplies $Z_a$, certifying $\widehat{S_a}$, matching the mirror ($b$-cycle-reduction) form of that theorem furnished by its symmetric statement in (I5). So $R(a)$ pairs with $E_b$ and $R(b)$ pairs with $E_a$. Hence
$$Pr(\A \text{ is not synch.}) \leq Pr(\neg\mathcal{R}(a) ; \neg\mathcal{R}(b)) + Pr(\mathcal{R}(a) ; \neg E_b) + Pr(\mathcal{R}(b) ; \neg E_a).$$
The first term is $O(1/n)$ by (I1). Exchanging the names of the two letters is a bijection of $\Sigma_n^2$ preserving the uniform measure, and it maps $\mathcal{R}(a) ; \neg E_b$ onto $\mathcal{R}(b) ; \neg E_a$. So the last two terms are equal, and it remains to show that
\begin{equation}
\label{eq:goal}
Pr(\mathcal{R}(a) ; \neg E_b) = O(1/n).
\end{equation}
No union bound over both $E_a$ and $E_b$ is needed here: on $\mathcal{R}(a)$, (I5) only ever asks for $E_b$.

\paragraph{Step 2: the chain.} Since $\mathcal{R}(a) \subseteq \mathcal{D}(st(a))$ by (I2),
\begin{equation}
\label{eq:split}
Pr(\mathcal{R}(a) ; \neg E_b) \leq Pr(\mathcal{D}(st(a)) ; \neg\mathcal{D}_b) + Pr(\mathcal{R}(a) ; \mathcal{D}_b ; \neg E_b).
\end{equation}
The first term is $O(1/n)$ by (I3)(i) (with $x=a$). For the second term, consider $\mathcal{R}(a) ; \mathcal{D}_b$. By (I2) the pair $st(a)$ is stable on $\mathcal{R}(a)$, and hence, by (I3)(iii), all pairs of $Z_b$ are stable. By (I4)(i),
$$\mathcal{R}(a) ; \mathcal{D}_b ; \neg E_b \ \subseteq\ \mathcal{D}_b ; \bigl(\neg Conn_b(Z_b) \vee BadCol_b(Z_b)\bigr).$$
The event on the right mentions only $b$ and $Z_b$. Remark~\ref{rem:chain}(a) followed by (I3)(ii) and (I4)(ii) give
$$Pr\bigl(\mathcal{D}_b ; (\neg Conn_b \vee BadCol_b)\bigr) \leq Pr(\neg Conn_b \mid \mathcal{D}_b) + Pr(BadCol_b \mid \mathcal{D}_b) = O(1/n).$$
Written as a single chain, in the form used throughout:
$$Pr(\mathcal{R}(a) ; \neg E_b) \leq \underbrace{Pr(\mathcal{D}(st(a)) ; \neg\mathcal{D}_b)}_{O(1/n)\ \text{by (I3)(i)}} + \underbrace{Pr(\mathcal{D}_b ; \text{(lemma property fails)})}_{\leq\ Pr(\text{property fails} \mid \mathcal{D}_b)\ =\ O(1/n)\ \text{by (I4)(ii)}}.$$
This proves \eqref{eq:goal}.

\paragraph{Step 3: conclusion.} By Step 1 and \eqref{eq:goal}, $Pr(\A \text{ is not synch.}) = O(1/n)$. Together with Lemma~\ref{lem_weak} this proves the theorem.
\end{proof}

\begin{myremark}
\emph{Why events are intersected and not conditioned on.} The events $\mathcal{R}(a)$ and `$st(a)$ is stable' depend on $b$ (whether the $1$-crown is reachable depends on $b$; cf.\ the footnote to Theorem~\ref{th_1stable}), although the pair $st(a)$ itself is defined by $a$ alone. Conditioning on $\mathcal{R}(a)$ would destroy the independence of $Z_b$ and $b$ on which Lemmas~\ref{lem_S_conn_clusters} and \ref{lem_stable_cluster} rely. In \eqref{eq:split} we therefore condition only on $\mathcal{D}_b$ and use $\mathcal{R}(a)$ solely to obtain the deterministic implication `pairs of $Z_b$ are stable'. The same device is used inside the proof of Theorem~\ref{thm:big_sync_sets} (Remark~\ref{rem:assumptions}).
\end{myremark}

\begin{myremark}
\emph{Reading (I4)(ii).} The proofs of Lemmas~\ref{lem_S_conn_clusters} and \ref{lem_stable_cluster} use, besides the independence of $Z$ and $y$, some facts about $y$ alone that hold with probability $1-O(1/n)$ (at most $5\ln n$ clusters by Lemma~\ref{lem_cnt_clusters}, $|\widehat{S_y}| = n - o(n)$). By (I2), (I3)(i) and $\mathcal{D}_b \subseteq \mathcal{D}(st(a))$ we have $Pr(\mathcal{D}_b) \geq 1 - \Theta(n^{-1/2}) - O(1/n) \geq 1/2$ for large $n$, so by Remark~\ref{rem:chain}(c) such a fact fails in $Pr(\cdot \mid \mathcal{D}_b)$ with at most twice its unconditional failure probability. This is why the lemmas may be applied in the conditioned space.
\end{myremark}

\section{Testing for Synchronization in Linear Expected Time}
\label{sec:algo}

In this section we show that following the
proof of Theorem~\ref{th_main} we can decide, whether or
not a given $n$-state automaton $\A$ is
synchronizing in linear expected time in $n$.
Notice that
the best known deterministic algorithm (basically due to
\v{C}ern\'y~\cite{Ce64}) for this problem is quadratic on the average and in the worst case.

\begin{theorem}
\label{th_opt_alg}
There is a deterministic algorithm for deciding whether or not a given $k$-letter automaton is synchronizing in linear expected time in $n$. Moreover, for $k > 1$, the proposed algorithm is optimal in expected time up to a constant factor.
\end{theorem}

\begin{proof}
First, we establish the lower bound for computational complexity. To do this, we make precise the computational model. We consider algorithms that query their input using questions of the form: ``What is the image of the state $q$ by the letter $\alpha$?'', denoted by $\text{query}(q,\alpha)$. Our lower bound results are stated in terms of the number of queries required to complete the computation. We do not take into account computation steps other than querying the input; this is standard for lower bounds since, in any classical representation of a deterministic automaton, queries are the standard mechanism to access transition data.

\begin{lemma}
\label{lem:alg_lower_bound}
Any deterministic algorithm that decides whether a given $n$-state $k$-letter automaton $\A$ is synchronizing (for $k > 1$) performs on average at least a linear number of queries in $n$.
\end{lemma}

\begin{proof}
It is necessary to verify that $\A$ has a unique SCC to claim that $\A$ is synchronizing. To ensure this, one must verify that there is no disconnected state (a state having connections only from itself; see Lemma~\ref{lem_weak}). To check this, an algorithm must verify that every state has either an incoming transition from another state or an outgoing transition to another state. Therefore, if an automaton is synchronizing, an algorithm must make at least $\lceil n/2 \rceil$ such queries to validate this property. Due to Theorem~\ref{th_main}, $\A$ has a unique SCC \whp. Thus, the average query complexity must be at least $\frac{n}{2}(1 - O(1/n)) = \Theta(n)$, which completes the proof of the lower bound.
\end{proof}

It follows from Lemma~\ref{lem:alg_lower_bound} that the linear expected time cannot be improved for any algorithm testing synchronization.

We now turn to describing the algorithm. The strategy is to sequentially verify whether all conditions established in Theorem~\ref{th_main} hold for $\A$. If they hold, we return `Yes'; otherwise, we fall back to the quadratic-time algorithm for $\A$. Since the probability that any of these conditions fails is $O(\frac{1}{n})$, the overall expected time remains linear in $n$, provided all conditions can be verified in linear time. Thus, it remains to prove that all conditions required in Theorem~\ref{th_main} can be verified in linear time.

We first describe the algorithm for $k=2$ and then explain how to generalize it for $k > 1$.

First, we execute Tarjan's linear-time algorithm~\cite{TarScc} on the underlying digraph of $\A$ to find minimal strongly connected components (MSCCs). If there are multiple MSCCs, we return `No' because $\A$ has multiple SCCs (and hence not synchronizing). Otherwise, there is a unique MSCC $\B$, and $\A$ is synchronizing if and only if $\B$ is synchronizing. Thus, all subsequent calculations are performed on the subautomaton $\B$. By Lemma~\ref{lem_size_of_subaut}, we may assume that $\B$ has at least $3n/4$ states (otherwise we fall back to the quadratic algorithm).

Now, fix a letter, say $a$, and consider its underlying transition graph. In this graph, each state $q \in Q$ belongs to a cluster $\text{cluster}(q)$. Let $\text{cycle}(q)$ denote the state with the smallest index on the cycle, and let $\text{cl}(q)$ denote the cycle length. Assuming cycle vertices are indexed in the clockwise direction from $0$ to $\text{cl}(q)-1$, let $\text{tree}(q)$ denote the tree in which $q$ is located, and let $\text{root}(q)$ denote the relative index of the root vertex on the cycle. Finally, let $\text{height}(q)$ denote the height of vertex $q$ in its tree ($\text{height}(q) = 0$ for cycle vertices). Both $\text{cluster}(q)$ and $\text{tree}(q)$ are labeled by the index of a representative state. We calculate the \emph{cluster structure} for each $q \in Q$, computing $\text{root}(q), \text{tree}(q), \text{cluster}(q), \text{cl}(q)$, and $\text{height}(q)$, alongside auxiliary information such as the total number and sizes of clusters and the unique highest $1$-branch (if it exists).

\begin{lemma}
\label{calc_lvls}
The cluster structure of a letter $x \in \Sigma$ can be calculated in linear time in $n$.
\end{lemma}

\begin{proof}
In each step, we choose an unobserved state $p \in Q$, set $\text{cluster}(p) = p$, and traverse the path
\begin{equation*}
p = p_0, p_1 = p_0.x, \dots, p_m = p_{m-1}.x
\end{equation*}
in the underlying digraph of $x$ until encountering a state $p_m$ such that $p_m = p_k = p_k.x^{m-k}$ for some $k < m$. When this occurs, we set $\text{cl}(p) = m-k$, set $\text{root}(p_i) = i-k$ for $k \le i < m$, and identify $\text{cycle}(p_m)$ as the cycle state with the smallest index.

Afterward, for each cycle state $q$, we run Breadth-First Search (BFS) in $\text{tree}(q)$ along reversed transitions. At step $j$, for a state $s$, we assign:
\begin{equation}
\text{height}(s) = j, \quad \text{tree}(s) = q, \quad \text{cluster}(s) = \text{cluster}(q), \quad \text{root}(s) = \text{root}(q), \quad \text{cl}(s) = \text{cl}(q).
\end{equation}
Note that we process the whole cluster in such iterations, so each state is observed at most twice in one iteration and across all subroutines, ensuring $O(n)$ time complexity. Simultaneously, we evaluate the total number of clusters and check for the existence of a unique highest $1$-branch satisfying Theorem~\ref{th_high_tree}.
\end{proof}

By Lemma~\ref{lem_cnt_clusters}, we may assume that the number of clusters does not exceed $5\ln n$. If a $1$-branch satisfying Theorem~\ref{th_high_tree} is found for one of the letters, we compute the highest $1$-branch in this tree in linear time by running the search subroutine restricted to this tree. Thus, using the cluster structure, we check in linear time whether letters $a$ and $b$ satisfy Theorem~\ref{th_high_tree}.

Due to letter independence, neither letter satisfies Theorem~\ref{th_joint_intersection}  with probability $O(\frac{1}{n})$. However, we perform the checks below for both letters in case both satisfy Theorem~\ref{th_joint_intersection} to follow the proof of Theorem~\ref{th_main}, falling back to quadratic algorithm only if verification fails for both letters, which happens with $O(1/n)$ probability due to Theorem~\ref{th_joint_intersection}. Note that since we operate on the MSCC, we have verified that the $1$-crown intersects with it for one of the letters.

Without loss of generality, let $a$ be a letter satisfying Theorem~\ref{th_joint_intersection}. For $a$ and its highest $1$-branch $T$, we locate a pair $\{r,q\}$, where $r$ is the root of $T$ and $q$ is the predecessor of the root of the tree containing $T$ on the $a$-cycle. The pair $\{r,q\}$ is stable by Theorem~\ref{th_1stable} and independent of $b$. Following the proof of Theorem~\ref{th_many_stable_ext}, we extend $\{r,q\}$ to a single set $Z_b$ of $\lceil n^{0.45} \rceil$ distinct stable pairs, independent for $b$ (this is the cheap, two-step construction of that theorem; the more expensive set independent of $a$ is never needed, by Theorem~\ref{thm:big_sync_sets}). Observing at most $O(n^{0.45})$ pairs, this step runs in linear time and fails with probability $O(\frac{1}{n})$.

Let $S_x$ be the set of clusters of $\Gamma_x$ containing more than $n^{0.95}$ states, and let $T_x$ be its complement. Construct the graph $\Gamma(S_b, Z_b)$ with vertex set $S_b$, placing an edge between two clusters if a pair from $Z_b$ connects them. Since $|S_b| \le 5\ln n$ and $|Z_b| \le n^{0.45} + 1$, we construct $\Gamma(S_b, Z_b)$ and verify its connectivity in linear time using Depth-First Search (DFS), producing a spanning tree of $\Gamma(S_b, Z_b)$. By Lemma~\ref{lem_S_conn_clusters}, connectivity fails with probability $O(\frac{1}{n})$.

Next, we calculate the greatest common divisor $d$ of the cycle lengths in $S_b$ using the Euclidean algorithm in $O(\ln^2 n)$ time. If $d > 1$, Lemma~\ref{lem_stable_cluster} requires verifying that no $g$-colouring (for $1 < g \mid d$) of $S_b$ preserves monochrome pairs under the action of $b$.

Let $\sigma$ be the partition on $Q$ defined by $p \sim_\sigma q \iff p.b^n = q.b^n$. For each cluster $C_i$ with cycle length $s_i$, states of $C_i$ partition into $s_i$ classes $C_{i,j}$ ($0 \le j \le s_i - 1$) such that $C_{i,j}.b \subseteq C_{i,(j+1 \bmod s_i)}$. Assuming $C_{i,0}$ contains a state $p$ with $\text{root}(p) = 0$, each class must be monochrome for some $1 < g \mid d \mid s_i$. For offsets $0 \le x_i \le d-1$, a state $q$'s color is:
\begin{equation*}
(g - x_i + \text{root}(q) - (\text{height}(q) \bmod d)) \bmod g.
\end{equation*}

For a pair $\{p,q\} \in Z_b$ with $p \in C_{i,\text{root}(p)}$ and $q \in C_{j,\text{root}(q)}$ to be monochrome for a $g$-colouring (for some $g \mid d$) requires:
\begin{equation}
\label{eqdiv}
g \mid (\text{root}(p) - \text{root}(q)) - (x_{\text{cluster}(p)} - x_{\text{cluster}(q)}) - (\text{height}(p) - \text{height}(q)).
\end{equation}
We first traverse the spanning tree of $\Gamma(S_b,Z_b)$ via DFS starting from root $r$ with $x_{\text{cluster}(r)} = 0$. Each edge corresponding to a pair $\{p,q\} \in Z_b$ uniquely determines the unassigned cluster offset $x_i \in \{0, \dots, d-1\}$.

Then we again compute the g.c.d. of $(\text{root}(p) - \text{root}(q)) - (x_{\text{cluster}(p)} - x_{\text{cluster}(q)}) - (\text{height}(p) - \text{height}(q))$ using Euclidean algorithm in $O(|Z_b|\ln{n}) = o(n)$ time.
If the g.c.d. is $1$, we effectively verified the requirement of Lemma~\ref{lem_stable_cluster}. The opposite holds with probability $O(\frac{1}{n})$ by Lemma~\ref{lem_stable_cluster}, in which case we call the fallback algorithm.

By this point we verified that all clusters of $\Gamma_b$ of size at least $n^{0.95}$ belong to one synchronizing set $\widehat{S_b}$. By Theorem~\ref{thm:big_sync_sets}, this is already enough; the symmetric check for $a$ (yielding $\widehat{S_a}$) is not needed.

It remains to check the pairwise cycle compatibility conditions from Theorem~\ref{thm:big_sync_sets}. Let the constant $\beta$ and $c_p, c_q, s_p, s_q$ be defined as in the theorem for two clusters. We now describe how we verify each scenario in the proof.

\begin{itemize}
\item \textbf{1. One cycle of a big size:} For each cycle of size $s_p \geq 2\beta$, we check whether half of its states belong to $\widehat{T}_b$. If that is the case for one of the clusters, we fall back to the quadratic algorithm (which happens with $O(1/n)$ probability). Clearly, this can be done in time linear in the number of cyclic states which is $O(n)$ (actually $O(\sqrt{n})$ \whp).

\item \textbf{2. Two cycles of big size ($2\beta \leq s_p \leq s_q$):} For each cycle $c_j$ of length more than $2\beta$, we compute cumulative counts $z_{j, i} \mid 1 \leq i \leq cl(c_j)$ of states that belong to $\widehat{T}_b$ starting from index $0$. This again can be done in time linear in the number of cyclic states. Then, for each pair of cycles $c_p, c_q$, we check that the sum of these counts $z_{p, s_p} + z_{q, s_p} \geq s_p$. If that's the case, we fall back to the quadratic algorithm, which again happens with probability $O(\frac{1}{n})$ due to the proof of that case.

\item \textbf{3. One of the cycles is small:} For this case to follow the proof, we just need to consider constant length expansion by letters and test whether some of the expanded states belong to $\widehat{T}_b$ or have collisions with previous states in expansion chains, which can easily be done in constant time and there are at most $O(\ln^2{n})$ cycle pairs that we have to run this procedure for. 

\end{itemize}

If all structural checks pass, we return `Yes' as we proved that this is sufficient for an automaton to be synchronizing. Otherwise, we execute the $O(n^2)$ fallback algorithm with probability $O(1/n)$, yielding an overall expected time of:
\begin{equation*}
O(n) + O(1/n) \cdot O(n^2) = O(n).
\end{equation*}

For an automaton $\A = \langle Q, \{a_1, a_2, \dots, a_k\} \rangle$ with $k > 2$, we run the algorithm on successive disjoint letter pairs $\A_i = \langle Q, \{a_{2i-1}, a_{2i}\} \rangle$. If any $2$-letter subautomaton is verified as synchronizing, we return `Yes'. If all $\lfloor k/2 \rfloor$ pairs fail, we execute the fallback algorithm on $\A$ in $O(n^2 k)$ time. Because letter assignments are independent, all pairs fail with probability $O(n^{-\lfloor k/2 \rfloor})$, yielding an expected running time of:
\begin{equation*}
O(n) + O(n^{-\lfloor k/2 \rfloor}) \cdot O(n^2 k) = O(n),
\end{equation*}
which is optimal for $k > 1$.
\end{proof}

\subsection{Experimental Verification}
\label{sec:experiments}

Pavel Ageev, a former master student of Mikhail Volkov, implemented a modified version of an
earlier version of the above algorithm~\cite{Ageev19}, relaxing the requirement that stable pairs (found
according to Section~\ref{sub_sec:extending_stable_pairs}) consist of pairwise distinct states,
a relaxation that does not affect correctness. His experiment measured the raw
fallback frequency of that implementation and found it consistent with $O(1/n)$.

We revisit this experimentally with two changes. First, rather than falling back to the
$O(n^2)$ algorithm the moment \emph{any} single structural check fails, our implementation
follows the version of the algorithm described in Section~\ref{sec:algo} that continues
through every non-critical check regardless of earlier failures, and only invokes the fallback
on the one \emph{critical} failure: the $1$-crown of neither letter intersects the (unique)
MSCC. This is a strictly weaker trigger than a blanket relax-nothing-else fallback, so the two
experiments are not measuring the same quantity, and a smaller empirical constant is expected.
Second, we report exact (Clopper--Pearson) $95\%$ confidence intervals throughout, rather than
single point estimates, and separately validate \emph{correctness} of the algorithm's optimistic
output (not just how often it invokes the fallback) against the ground-truth quadratic check.
Our implementation, experiment scripts, and raw results are available at
\url{https://github.com/berlm/synchronizable-automata}.

\paragraph{Critical-failure rate.} For each $n \in \{50, 100, 200, 400, 800, 1600, 3200, 6400,
12800, 25600\}$ we generated independent uniformly random binary automata and ran the algorithm without computing ground truth, recording whether the critical failure occurred. The number of trials was tapered from $5000$ (small $n$) down to $400$ (largest $n$) to bound total running time while keeping the confidence intervals informative at both ends of the range.

\begin{center}
\begin{tabular}{|l|r|r|r|r|r|r|r|r|r|r|}
\hline
$n$ & 50 & 100 & 200 & 400 & 800 & 1600 & 3200 & 6400 & 12800 & 25600 \\ \hline
trials & 5000 & 5000 & 5000 & 4000 & 3000 & 2500 & 2000 & 1200 & 700 & 400 \\ \hline
critical failures & 217 & 109 & 61 & 22 & 4 & 5 & 4 & 1 & 0 & 0 \\ \hline
rate ($\times 10^{-3}$) & 43.4 & 21.8 & 12.2 & 5.5 & 1.33 & 2.00 & 2.00 & 0.83 & 0 & 0 \\ \hline
$n \times$ rate & 2.17 & 2.18 & 2.44 & 2.20 & 1.07 & 3.20 & 6.40 & 5.33 & 0 & 0 \\ \hline
\end{tabular}
\end{center}

A log-log regression of the rate against $n$, restricted to the eight sizes with at least one
observed failure, gives an exponent of $-0.80 \pm 0.21$ ($95\%$ CI), consistent with the
theoretical prediction of exactly $-1$. Pooling all ten sizes under the model
$P(\text{critical failure}) = \alpha/n$ (weighting each size by its number of trials) gives
$\alpha = 2.21 \pm 0.21$ ($95\%$ CI $[2.00, 2.42]$).

As another consistency check, the individual structural statistics that feed into the algorithm were tracked separately at $n=1600$ and $n=6400$: the fraction of runs in which a single letter's $1$-crown failed to intersect the MSCC was $0.025$--$0.028$ at $n=1600$ and $0.015$--$0.016$ at $n=6400$, matching $1/\sqrt n$ (respectively $0.025$ and $0.0125$) to within sampling noise, as predicted by Theorem~\ref{th_joint_intersection}.

\section{Conclusions}
Theorem~\ref{th_main} gives an exact order of the
convergence rate for the probability of being
synchronizable for $2$-letter automata
up to a constant factor. It should be fairly easy to adapt the current proof technique to establish the convergence rate for the $k$-size alphabet case ($k>1$) to be $1-O(\frac{1}{n^{0.5k}})$, because the main restriction comes from the probability of having a unique $1$-branch for some letter. Thus, perhaps the most natural open question here is about the tightness of the convergence rate $1-O(\frac{1}{n^{0.5k}})$ for the $k$-letter alphabet case. We predict a tighter upper bound $1-\Theta(\frac{1}{n^{k-1}})$ for general random $k$-letter automata (for $k>2$) matching the lower bound for automata having multiple terminal SCCs.

Since only automata with a unique terminal SCC can be synchronizing, the second natural open question is about the convergence rate of the probability that a random automaton with a unique SCC is synchronizable for the uniform distribution. Especially, the binary alphabet is of certain interest because the upper bound for this case comes from the case of multiple SCCs. We predict an exponentially small probability of not being synchronizable for the automata with unique SCCs. Additionally, it might be of interest to explore particular cases of automata that appear in practical applications, such as prefix-code automata.

Another challenging problem concerns a generalization of the synchronization property. Given $d \geq 1$, what is the probability that for an automaton chosen uniformly at random the minimum rank of the automaton (over all possible words) is equal to $d$? Notice that the case $d = 1$ corresponds to the original problem of synchronization, and for $d = n$ it is the probability that all letters are permutations, which is $\left({n!}/{n^{n}}\right)^{k} = \Theta(n^{k/2} e^{-nk})$.
We believe that for $d>1$ the probability is also exponentially close to $0$ for automata with unique SCC and, if one could prove it, this would also lead to the positive answer to the aforementioned hypothesis for the automata with unique SCC.

We have also established an algorithm that checks whether a given finite automaton is synchronizing in linear expected time. Note that unlike in~\cite{FastSyn}, neither our proof nor our algorithm yields a synchronizing word. Even though in the worst case the length of a reset word can be at least quadratic in $n$, it doesn't rule out the possibility of an algorithm with a linear  expected time yielding a reset word (if it exists), given that~\cite{ChapuyPerarnau2023} established that a random synchronizing automaton has a reset word of length $O(\sqrt{n}\ln{n})$.

\section{Acknowledgments}

The author is grateful to Mikhail Volkov for constant support in the research and also to Cyril Nicaud, Marek Szyku{\l}a, Dominique Perrin, Marie-Pierre B{\'e}al, Pavel Ageev and Julia Mikheeva for their interest and useful suggestions.
The author is also thankful to anonymous referees of Theoretical Computer Science journal for very useful feedback which helped to improve the presentation of the results. A significant part of the work has been done by the author while being employed by the Institute of Mathematics and Natural Sciences, Ural Federal University, Ekaterinburg, Russia.

During the preparation of this work, the author used AI models (Gemini 1.5 Pro and Claude Sonnet 5) and an internal Google publishing review tool to find any issues and polish the presentation of the results. The author reviewed and edited the content to the best of his abilities and takes full responsibility for the content of the publication.
\bibliography{main}

\end{document}